\documentclass[sort&compress]{elsarticle}

\usepackage[raiselinks=false,colorlinks=true,citecolor=blue]{hyperref}

\journal{Theoretical Computer Science}

\DeclareFontFamily{U}{stmry}{}
\DeclareFontShape{U}{stmry}{b}{n}
   {  <5> <6> <7> <8> <9> <10> gen * stmary
      <10.95><12><14.4><17.28><20.74><24.88>stmary10%
   }{}
\DeclareFontShape{U}{stmry}{m}{n}
   {  <5> <6> <7> <8> <9> <10> gen * stmary
      <10.95><12><14.4><17.28><20.74><24.88>stmary10%
   }{}

\usepackage{xcolor}
\usepackage{amsmath,amssymb,amsfonts,amsthm}
\newtheorem{theorem}{Theorem}[section]

\newtheorem{lemma}[theorem]{Lemma}
\newtheorem{corollary}[theorem]{Corollary}
\newdefinition{example}[theorem]{Example}
\newdefinition{definition}[theorem]{Definition}
\newdefinition{remark}[theorem]{Remark}
\usepackage{xspace}
\usepackage{proof}
\usepackage[only,llbracket,rrbracket]{stmaryrd}
\usepackage{bbm}

\newcommand\FV{\ensuremath{\mathsf{FV}}}
\newcommand\proofof[1]{\noindent\textcolor{darkgray}{\sffamily\bfseries Proof of #1.}}
\newcommand\titre[1]{\noindent\textcolor{darkgray}{\rm\sffamily\bfseries (#1)}}

\newcommand\s[1]{\ensuremath{s({#1})}}
\newcommand\conj[1]{\ensuremath{\mathsf{conj}(#1)}}
\newcommand\PF[1]{\ensuremath{\mathsf{PF}(#1)}}
\newcommand\OC{System I\xspace}
\newcommand\OCe{System I$^\eta$\xspace}

\newcommand\interp[1]{\ensuremath{\left\llbracket{#1}\right\rrbracket}}

\newcommand\SN{\ensuremath{\mathsf{SN}}}

\newcommand\eq{\ensuremath{\rightleftarrows}}
\newcommand\re{\ensuremath{\hookrightarrow}}
\newcommand\toreq{\ensuremath{\rightsquigarrow}}
\newcommand\basicr{\toreq_{\triangle}}
\newcommand\basicre{\re_{\triangle}}

\newcommand\ve[1]{#1}
\newcommand\cond[1]{\ensuremath{\scriptstyle[#1]\,}}
\newcommand\condi[2]{\ensuremath{{#1}{\vcenter{#2}}}}
\newcommand\rulelabel[1]{\mbox{\scriptsize\sc{#1}}}

\newcommand\V{\mathcal V}

\begin{document}

\begin{frontmatter}

  \title{Extensional proofs\\
    in a propositional logic modulo isomorphisms\tnoteref{funding}}
    \tnotetext[funding]{Partially funded by 
    PIP 11220200100368CO,
    PICT 2019-1272 and 2021-I-A-00090, 
    CSIC 22520220100073UD,
    and the French-Argentinian IRP SINFIN.}

\author[ICC,UNQ]{Alejandro D\'iaz-Caro}
\author[LMF]{Gilles Dowek}

\address[ICC]{ICC, CONICET--Universidad de Buenos Aires, Argentina}
\address[UNQ]{Departamento de Ciencia y Tecnolog\'ia, Universidad Nacional de Quilmes, Argentina}
\address[LMF]{Inria, LMF, ENS Paris-Saclay, Gif-sur-Yvette, France}

\begin{abstract}
\OC is a proof language for a fragment of propositional logic where isomorphic
propositions, such as $A\wedge B$ and $B\wedge A$, or $A\Rightarrow(B\wedge C)$
and $(A\Rightarrow B)\wedge(A\Rightarrow C)$ are made equal. \OC enjoys the
strong normalization property. This is sufficient to prove the existence of
empty types, but not to prove the introduction property (every closed term in
normal form is an introduction). Moreover, a severe restriction had to be made
on the types of the variables in order to obtain the existence of empty types.
We show here that adding $\eta$-expansion rules to \OC permits to drop this
restriction, and yields a strongly normalizing calculus which enjoys the full
introduction property.
\end{abstract}


\end{frontmatter}

\section{Introduction}

\subsection{Making connectives algebraic}

Operations in mathematics are often associative, commutative, equipped
with a neutral element, etc.  In contrast, the logical connectives
have no algebraic properties.  Of course, if the proposition $A \wedge
B$ has a proof, then so does the proposition $B \wedge A$, but if $\ve
r$ is a proof of $A \wedge B$, then it is not a proof of $B \wedge
A$. Thus, if we consider two propositions equal when they have the
same proofs, the propositions $A \wedge B$ and $B \wedge A$ are
different. This lack of algebraic properties of the connectives
distinguishes logic, among the mathematical theories.

Our project is to bring logic closer to algebra, by making, for example, the
conjunction commutative, that is the propositions $A \wedge B$ and $B
\wedge A$ equal. This extends the project of Martin-L\"of's type
theory \cite{MartinLof84}, the Calculus of Constructions
\cite{CoquandHuetIC88}, Deduction modulo theory
\cite{DowekHardinKirchnerJAR03,DowekWernerJSL98}, etc.  that makes
definitionally equivalent propositions equal.

The propositions $A \wedge B$ and $B \wedge A$ already are equal in
some models: in Boolean algebras, or in Heyting algebras, conjunction
and disjunction are associative and commutative, they are distributive
one over the other, etc. Thus, they are genuine algebraic
operations. In categorical models, the Cartesian product is
associative and commutative, but only modulo isomorphisms. And, on the
syntactic side, conjunction and disjunction are neither associative
nor commutative.

Our long term objective is to understand how identifying some
propositions impacts proof-theory: how must the notion of
proof-reduction, that of reducibility candidate, etc.  be modified.
To explore such questions, we start with a simple case: constructive
propositional logic with implication and conjunction.

\subsection{Logical isomorphism}

The first step in such a project is to understand which propositions
can be identified. An answer to this question is given by the notion of
logical isomorphism. Two propositions $C$ and $D$ are said to be
isomorphic when there exist proofs of $C \Rightarrow D$ and $D
\Rightarrow C$, whose composition, in both ways, is semantically
equivalent to the identity.  For instance, the propositions $A \wedge
B$ and $B \wedge A$ are isomorphic.

This notion of isomorphism has been studied by Solov'ev~\cite{SolovievJSM83}, in the context of category theory, in the calculus with conjunction, implication, and a constant true type $\top$. Solov'ev shows that all the realizations of formal combinations of objects using the functors corresponding to the conjunction (cartesian product) and the implication (hom) are isomorphic in all Cartesian closed categories if and only if some of their realizations in the category of finites sets are isomorphic. 
A syntactic decision algorithm for isomorphisms is also given, with an efficient algorithm in~\cite{GilZibinMSCS05}.

Latter, the notion of isomorphism was also studied by M.~Rittri
\cite{RittriCADE90}, who has shown that identifying isomorphic
propositions simplified the search for a lemma in a database of
mathematical results. Then, such isomorphisms, for different
constructive systems, have been characterized by K.~Bruce, G.~Longo,
and R.~Di Cosmo
\cite{BruceDiCosmoLongoMSCS92,DiCosmo95,DiCosmoMSCS05}. O.~Laurent has
then extended this characterization to classical logic
\cite{Laurent05b}.

In the case of constructive propositional logic with implication and
conjunction four isomorphisms can be considered.
\begin{align*}
A \wedge B                 &\equiv B \wedge A\\
A \wedge (B \wedge C)      &\equiv (A \wedge B) \wedge C\\
A \Rightarrow (B \wedge C) &\equiv (A \Rightarrow B) \wedge (A \Rightarrow C)\\
(A \wedge B) \Rightarrow C &\equiv A \Rightarrow B \Rightarrow C
\end{align*}

\subsection{Non-deterministic proof-reduction}

Another question that arises in such a project is that of the
determinism of proof-reduction. The first models of computations:
Turing machines, $\lambda$-calculus, etc. were often deterministic.
But, quickly, some non-deterministic variants were introduced.  This
non-determinism then became essential with the rise of quantum
computing and asynchronous parallel computing
\cite{BoudolIC94,BucciarelliEhrhardManzonettoAPAL12,deLiguoroPipernoIC95,DezaniciancagliniDeliguoroPipernoSIAM98,PaganiRonchidellaroccaFI10}.

In proof-languages, in contrast, the reduction is still deterministic.
Yet, there are several situations where non-determinism is natural.
For example, if we diagonalize
the conjunction, introducing a unary connective $\hat{\wedge}$ such
that $\hat{\wedge} A = A \wedge A$, then the introduction rule of the
conjunction becomes
$$\infer[\mbox{$\hat{\wedge}$-i}]
{\hat{\wedge} A}{A~~~A}$$
and its first elimination rule
$$\infer[\mbox{$\hat{\wedge}$-e1}]{A}{\hat{\wedge} A}$$
Then, the proof 
$$\infer[\mbox{$\hat{\wedge}$-e1}]{A}{ \infer[\mbox{$\hat{\wedge}$-i}]
  {\hat{\wedge} A}{\infer{A}{\pi_1}~~~\infer{A}{\pi_2}}}$$ reduces to
$\pi_1$. But, thanks to the diagonalization, it can also be reduced to
$\pi_2$.  As we shall see, making conjunction commutative introduces
non-determinism in a similar way.

\subsection{\texorpdfstring{\OC}{System I}}

\OC \cite{DiazcaroDowekFSCD19} is a first attempt to identify
isomorphic propositions in constructive propositional logic with
implication and conjunction.

The usual proof-language of this logic is simply typed lambda-calculus with
Cartesian product. In this calculus, the term $\lambda x^A.r \times \lambda x^A.
s$, where we write $u \times v$ for the pair of two terms $u$ and $v$, has type
$(A \Rightarrow B) \wedge (A \Rightarrow C)$. In \OC, as $(A \Rightarrow B)
\wedge (A \Rightarrow C)\ \equiv\ A \Rightarrow (B \wedge C)$, this term also
has type $A \Rightarrow (B \wedge C)$ and it can be applied to $\ve t$ of type
$A$, yielding the term $(\lambda x^A.r \times \lambda x^A.s) \ve t$ of type $B
\wedge C$. With the usual reduction rules of lambda-calculus with pairs, such a
mixed cut (an introduction followed by the elimination of another connective)
would be in normal form, but we also extended the reduction relation, with an
equation $(\lambda x^A.r \times \lambda x^A.s) \eq \lambda x^A.(r \times s)$,
so that this term can be
$\beta$-reduced, 
taking inspiration from rules well-known in the area
of program transformation, for instance in 
G.~R\'ev\'esz \cite{Revesz1992,Revesz1995}, K.~St{\o}vring
\cite{Stovring06extendingthe}, and others.

One of the difficulties in the design of \OC was the definition of the
elimination rule for the conjunction. We cannot use a rule like ``if $\ve
r:A\wedge B$ then $\pi_1(\ve r):A$''. Indeed, if $A$ and $B$ are two arbitrary
types, $\ve s$ a term of type $A$ and $\ve t$ a term of type $B$, then $\ve
s\times \ve t$ has both type $A\wedge B$ and type $B \wedge A$, thus $\pi_1(\ve
s \times \ve t)$ would have both type $A$ and type $B$. The solution is to
consider explicitly typed (Church style) terms, and parameterize the projection
by the type: if $\ve r:A\wedge B$ then $\pi_A(\ve r):A$ and the reduction rule
is then that $\pi_A(\ve s \times \ve t)$ reduces to $\ve s$ if $\ve s$ has type
$A$. Thus, $\pi$-reduction is type driven, and $\beta$-reduction as well.

This rule makes reduction non-deterministic. Indeed, in the particular
case where $A$ is equal to $B$, then both $\ve s$ and $\ve t$ have
type $A$ and $\pi_A(\ve s \times \ve t)$ reduces both to $\ve s$ and
to $\ve t$. Unlike in the lambda-calculus we cannot specify which
reduct we get, but in any case, we eventually get a term in normal form
of type $A$, that is a cut-free proof of $A$.  Therefore, \OC is a
non-deterministic calculus and our pair-construction operator $\times$
is also the parallel composition operator of a non-deterministic
calculus. More precisely, the non-determinism does not come from one
operator, but from the interaction of two operators, $\times$ and
$\pi$. In this respect, \OC is close to quantum and algebraic
$\lambda$-calculi~\cite{ArrighiDiazcaroLMCS12,ArrighiDiazcaroValironIC17,ArrighiDowekRTA08,
  ArrighiDowekLMCS17,VauxMSCS09,DiazcaroPetitWoLLIC12,DiazcaroDowekTPNC17,DiazcaroGuillermoMiquelValironLICS19}
where the non-determinism comes from the interaction of superposition
and projective measurement.

In~\cite{DiazcaroMartinezlopezIFL15}, we have implemented an early
version of System I, extended with general recursion. We showed with a
couple of examples, how this language can be helpful as a realistic
programming language. On the one hand, the language has a sort of
partial application, which can start the computation as soon as it
receives a parameter, in any order. On the other hand, the language
enables to reuse code by projecting functions and discarding unused
code prior to its usage. For example, a function calculating the
quotient and the rest of two natural numbers can be projected out into
a function calculating only the quotient, discarding the code
calculating the rest. Even in the case of general recursion and mutual
recursion, the language will unfold the recursion as needed to discard
the unused code.

\subsection{The drawbacks of \texorpdfstring{\OC}{System I}}

In \cite{DiazcaroDowekFSCD19} we succeeded in proving the strong
normalization and the consistency of \OC, that is, the existence of a
proposition that has no closed proof.

However, \OC still has some drawbacks.
\begin{itemize}
\item As the propositions $A \Rightarrow B \Rightarrow A$ and $B \Rightarrow A
  \Rightarrow A$ are isomorphic, the term $(\lambda x^A.\lambda y^B.x)r$ where $r$
  has type $B$ is well-typed, but it cannot be $\beta$-reduced. In \OC, this term
  is in normal form, so \OC does not verify the introduction property (a closed
  term in normal form is either an abstraction of a pair). Only when such a term
  is applied to a term $s$ of type $A$, to make a closed term of atomic type, it
  can be reduced: $(\lambda x^A.\lambda y^B.x)r s$, being equivalent to $(\lambda
  x^A.\lambda y^B.x) s r$, can be reduced to $(\lambda y^B.s) r$, and then to $s$.
  A solution has been explored in \cite{DiazcaroMartinezlopezIFL15}: ``delayed
  $\beta$-reduction'' that reduces $(\lambda x^A.\lambda y^B.x)r$ to $\lambda
  x^A.(\lambda y^B.x) r$ and then to $\lambda x^A.x$. A similar equivalence has
  been proposed before in the context of proof-nets~\cite{RegnierTCS94}.

\item As the types $(A \wedge B) \Rightarrow (A \wedge B)$ and $A \Rightarrow B
  \Rightarrow (A \wedge B)$ are isomorphic, the term $(\lambda x^{A \wedge B}.x)r$
  where $r$ has type $A$ is well-typed (of type $B \Rightarrow (A \wedge B)$), but
  it cannot be $\beta$-reduced as the term $r$ of type $A$, cannot be substituted
  for the variable $x$ of type $A \wedge B$. In \OC variables have so called
  ``prime types'', that is, types that do not contain a conjunction at head
  position. Thus, the above term can only be written as $(\lambda y^{A}.\lambda
  z^B.y \times z)r$, and it reduces to $\lambda z^B.r \times z$. Another
  possibility has been explored in \cite{DiazcaroMartinezlopezIFL15}: ``partial
  $\beta$-reduction'' that reduces directly $(\lambda x^{A \wedge B}.x)r$ to
  $\lambda z^B.r \times z$.
  It is interesting to remark that a similar notion to prime types has been already proposed by Solov'ev in~\cite{SolovievJSM83}.
\end{itemize}

\subsection{\texorpdfstring{\OCe}{System I-eta}}

In this paper, we show these two drawbacks are symptoms of the lack of
extensionality in \OC. This leads us to introduce the \OCe that
extends \OC with an $\eta$-expansion rule, and a surjective pairing
$\delta$-expansion rule.

In \OCe, the term $(\lambda x^A.\lambda y^B.x)r$ $\eta$-expands to $\lambda
x^A.(\lambda x^A.\lambda y^B.x)r x$, that is equivalent to $\lambda x^A.(\lambda
x^A.\lambda y^B.x) x r$, and reduces to $\lambda x^A.x$. In the same way, the
term $(\lambda x^{A \wedge B}.x)r$ $\eta$-expands to $\lambda y^B.(\lambda x^{A
\wedge B}.x)r y$, that is equivalent to $\lambda y^B.(\lambda x^{A \wedge
B}.x)(r \times y)$, and reduces to $\lambda y^B.r \times y$. This way, we do not
need to constrain variables to have prime types.
Dropping this restriction, makes the mixed cut $\pi_{(\tau \wedge \tau)
\Rightarrow \tau}(\lambda x^{\tau \wedge \tau}.x)$ well-typed, since $(A\wedge
B)\Rightarrow C$ is isomorphic to $A\Rightarrow B\Rightarrow C$ and variables
can have any type. However, using the $\delta$-rule this term expands to
$\pi_{(\tau \wedge \tau) \Rightarrow \tau}(\lambda x^{\tau \wedge
\tau}.\pi_{\tau} (x) \times \pi_{\tau}(x))$ that is equivalent to $\pi_{(\tau
\wedge \tau) \Rightarrow \tau} ((\lambda x^{\tau \wedge \tau}.\pi_{\tau} (x))
\times (\lambda x^{\tau \wedge \tau}.\pi_{\tau} (x)) )$, and reduces to $\lambda
x^{\tau \wedge \tau}.\pi_{\tau} (x)$ that is an introduction.

Designing \OCe yet led us to make a few choices. For instance, if the terms $r$
and $s$ are not introductions, then $(r \times s) t$, where $t$ has type $A$,
$\eta$-expands to $(\lambda x^A.(r x) \times \lambda x^A. (s x)) t$, that is
equivalent to $\lambda x^A.((r x) \times (s x)) t$ and $\beta$-reduces to $(r t)
\times (s t)$. But, if one of them is an abstraction on a type different from
$A$, then the term cannot be reduced. For instance $((\lambda x^{\tau
\Rightarrow \tau}.\lambda y^{\tau}.x) \times (\lambda y^{\tau}.y))t$, where $t$
is a term of type $\tau$, cannot be reduced. So we could either introduce a
symmetric rule to commute the two abstractions or introduce a distributivity
rule transforming the elimination $((\lambda x^{\tau \Rightarrow \tau}.\lambda
y^{\tau}.x) \times(\lambda y^{\tau}.y))t$ into the introduction $(\lambda
x^{\tau \Rightarrow \tau}.\lambda y^{\tau}.x) y \times (\lambda y^{\tau}.y)t$.
We have chosen the second option, as we favoured reduction over equivalence. But
both choices make sense.

Our main results are the normalization proof of \OCe, developing ideas from
\cite{DiazcaroDowekFSCD19,JayGhaniJFP95} and the introduction property,
showing that \OCe solves the problems of \OC.

\section{Type isomorphisms}\label{sec:eqTypes}

We first define the types and their equivalence, and state properties on this
relation. Some of these properties have been proved
in~\cite{DiazcaroDowekFSCD19}, and others are new.

\subsection{Types and isomorphisms}
Types are defined by the following grammar, where $\tau$ is the only atomic
type, $\Rightarrow$ is the constructor of the type of functions, and $\wedge$ is
the constructor of the type for pairs.
\[
  A\ =\ \tau~|~A\Rightarrow A~|~A\wedge A
\]

\begin{definition}[Size of a type]
  The size of a type is defined as usual by
  \begin{align*}
    \s{\tau} &= 1\\ 
    \s{A\Rightarrow B} &=\s{A}+\s{B}+1\\
    \s{A\wedge B} &=\s{A}+\s{B}+1
  \end{align*}
\end{definition}

\begin{definition}[Type equivalence]
  The equivalence between types is the smallest congruence such that:
  \begin{align*}
    A\wedge B		  &\equiv B\wedge A \\
    A\wedge (B\wedge C)	  &\equiv (A\wedge B)\wedge C \\
    A\Rightarrow(B\wedge C)  &\equiv (A\Rightarrow B)\wedge(A\Rightarrow C) \\
    (A\wedge B)\Rightarrow C &\equiv A\Rightarrow B\Rightarrow C
  \end{align*}
\end{definition}

\begin{remark}
This equivalence relation is decidable \cite[Theorem 6.4.5]{DiCosmo95}
as these equivalences can be oriented as rewrite rules yielding a
normal form modulo associativity and commutativity.  We also have
defined a notion of canonical form
in~\cite{DiazcaroMartinezlopezIFL15} to implement an earlier version
of \OC.  However, as different orientations can be chosen for
distributivity and curryfication we prefer, in this theoretical
presentation, to give the typing rules for the equivalence relation
and not for a specific choice of a canonical form.
\end{remark}

\subsection{Prime factors}
We recall a lemma proved in \cite{DiazcaroDowekFSCD19} stating that any type is
equivalent to a conjunction of prime types \cite{DiazcaroDowekFSCD19}.

This transformation of a type into a conjunction of prime types can be compared
to the transformation of a proposition as a conjunction of clauses, except that
we use the equivalence $\equiv$ and not logical equivalence.

\begin{definition}[Prime types]
  A prime type is a type of the form $C_1\Rightarrow\dots\Rightarrow
  C_n\Rightarrow\tau$, with $n\geq 0$.
\end{definition}

A prime type is equivalent to $(C_1\wedge\cdots\wedge C_n)\Rightarrow\tau$,
which is either equivalent to $\tau$ or to $C\Rightarrow\tau$, for some $C$. For
uniformity, we may write $\varnothing\Rightarrow\tau$ for $\tau$.
We prove that each type can be decomposed into a conjunction of prime types.
We use the notation $[A_i]_{i=1}^n$ for the multiset whose elements are
$A_1,\dots,A_n$, we write $\uplus$ for the union of multisets, and we write
$\conj{[A_i]_{i=1}^n}$ for $A_1\wedge\cdots\wedge A_n$. We write
$[A_1,\dots,A_n]\sim[B_1,\dots,B_m]$ if $n=m$ and $B_i\equiv A_i$.

\begin{definition}[Prime factors]
  The multiset of prime factors of a type $A$ is inductively defined as follows,
  with the convention that $A\wedge\varnothing=A$.
  \begin{align*}
    \PF\tau &=[\tau]\\
    \PF{A\Rightarrow B} &= [(A\wedge B_i)\Rightarrow\tau]_{i=1}^n \quad\textrm{ where }[B_i\Rightarrow\tau]_{i=1}^n=\PF B\\
    \PF{A\wedge B} &=\PF A\uplus\PF B
  \end{align*}
\end{definition}

\begin{lemma}\label{lem:eqConjPF}
  For all $A$, $A\equiv\conj{\PF A}$.
\end{lemma}
\begin{proof}
 By induction on \s A.
 \begin{itemize}
 \item If $A=\tau$, then $\PF\tau=[\tau]$, and so $\conj{\PF\tau}=\tau$.
 \item If $A=B\Rightarrow C$, then $\PF A = [(B\wedge C_i)\Rightarrow\tau]_i$,
   where $[C_i\Rightarrow\tau]_i=\PF C$. By the i.h., $C\equiv\bigwedge_i(C_i\Rightarrow\tau)$,
   hence, $A = B\Rightarrow C\equiv B\Rightarrow\bigwedge_i (C_i\Rightarrow\tau)\equiv
   \bigwedge_i(B\Rightarrow C_i\Rightarrow\tau)\equiv\bigwedge_i((B\wedge C_i)\Rightarrow\tau)$.
 \item If $A=B\wedge C$, then $\PF{A}=\PF{B}\uplus\PF{C}$. By the i.h.,
   $B\equiv\conj{\PF B}$,
   and $C\equiv\conj{\PF C}$.
   Therefore,
   $A
   =B\wedge C
   \equiv \conj{\PF B}\wedge\conj{\PF C}\equiv\conj{\PF{B\wedge C}}
   \equiv\conj{\PF B\uplus\PF C}
   =\conj{\PF A}$.
   \qedhere
 \end{itemize}
\end{proof}

\begin{lemma}\label{lem:eqPF}
  If $A\equiv B$, then $\PF A\sim\PF B$.
\end{lemma}
\begin{proof}
  First we check that $\PF{A\wedge B}\sim\PF{B\wedge A}$ and similar for the
  other three isomorphisms. Then we prove by structural induction that if $A$ and
  $B$ are equivalent in one step, then $\PF A\sim\PF B$. We conclude by an
  induction on the length of the derivation of the equivalence $A\equiv B$.
\end{proof}

\subsection{Measure of types}

The size of a type is not preserved by equivalence. For instance,
$\tau\Rightarrow(\tau\wedge\tau) \equiv
(\tau\Rightarrow\tau)\wedge(\tau\Rightarrow\tau)$, but
$\s{\tau\Rightarrow(\tau\wedge\tau)}=5$ and
$\s{(\tau\Rightarrow\tau)\wedge(\tau\Rightarrow\tau)}=7$. Thus, we define
another notion of measure of a type, conforming the usual relation.

\begin{definition}[Measure of a type]\label{def:measure}
  The measure of a type is defined as follows
  \[
    m(A) = \sum_i(m(C_i)+1)\quad\textrm{ where }[C_i\Rightarrow\tau]=\PF A
  \]
  with the convention that $m(\emptyset)=0$.
\end{definition}

The following lemma states that the given measure conforms the usual relation.
\begin{lemma}\label{lem:comp}~
  \begin{enumerate}
  \item $m(A\wedge B)>m(A)$
  \item $m(A\Rightarrow B) >m(A)$
  \item $m(A\Rightarrow B) > m(B)$
  \item if $A\equiv B$, $m(A)=m(B)$
  \end{enumerate}
\end{lemma}
\begin{proof}~
  \begin{enumerate}
  \item $\PF A$ is a strict submultiset of $\PF{A\wedge B}$.
  \item Let $\PF B=[C_i\Rightarrow\tau]_{i=1}^n$. Then, $\PF{A\Rightarrow
      B}=[(A\wedge C_i)\Rightarrow\tau]_{i=1}^n$. Hence, $m(A\Rightarrow B)\geq
    m(A\wedge C_1)+1>m(A\wedge C_1)\geq m(A)$.
  \item $m(A\Rightarrow B) = \sum_i m(A \wedge C_i) + 1 > \sum_i m(C_i) + 1 =
    m(B)$.
  \item By induction on $s(A)$. Let $\PF A = [C_i\Rightarrow\tau]_i$ and $\PF B
    = [D_j\Rightarrow\tau]_j$. By Lemma~\ref{lem:eqPF},
    $[C_i\Rightarrow\tau]_i\sim[D_i\Rightarrow\tau]_i$. Without lost of generality,
    take $C_i\equiv D_i$. By the induction hypothesis, $m(C_i)=m(D_i)$. Then,
    $m(A)=\sum_i(m(C_i)+1)=\sum_i(m(D_i)+1)=m(B)$.
    \qedhere
  \end{enumerate}
\end{proof}

\subsection{Decomposition properties on types}
In simply typed lambda calculus, the implication and the conjunction are
constructors, that is $A\Rightarrow B$ is never equal to $C\wedge D$, if
$A\Rightarrow B=A'\Rightarrow B'$, then $A=A'$ and $B=B'$, and the same holds
for the conjunction. This is not the case in \OCe, where $\tau\Rightarrow
(\tau\wedge\tau)\equiv (\tau\Rightarrow\tau)\wedge(\tau\Rightarrow\tau)$, but
the connectors still have some coherence properties, which are proved in this
section.

\begin{lemma}\label{lem:ImpConj}
  If $A\Rightarrow B\equiv C_1\wedge C_2$, then $C_1\equiv A\Rightarrow B_1$ and
  $C_2\equiv A\Rightarrow B_2$ where $B\equiv B_1\wedge B_2$.
\end{lemma}
\begin{proof}
  By Lemma~\ref{lem:eqPF}, $\PF{A\Rightarrow B}\sim\PF{C_1\wedge
    C_2}=\PF{C_1}\uplus\PF{C_2}$. Let $\PF B=[D_i\Rightarrow\tau]_{i=1}^n$, so
  $\PF{A\Rightarrow B}=[( A\wedge D_i)\Rightarrow\tau]_{i=1}^n$. Without lost of
  generality, take $\PF{C_1}\sim[(A\wedge D_i)\Rightarrow\tau]_{i=1}^k$ and
  $\PF{C_2}\sim[(A\wedge D_i)\Rightarrow\tau]_{i=k+1}^n$. Therefore, by
  Lemma~\ref{lem:eqConjPF}, we have $A\Rightarrow
  B\equiv\bigwedge_{i=1}^k((A\wedge{D_i})\Rightarrow\tau)\wedge\bigwedge_{i=k+1}^n((A\wedge{D_i})\Rightarrow\tau)\equiv
  (A\Rightarrow\bigwedge_{i=1}^k
  (D_i\Rightarrow\tau))\wedge(A\Rightarrow\bigwedge_{i=k+1}^n
  (D_i\Rightarrow\tau))$. Take $B_1=\bigwedge_{i=1}^k{D_i}\Rightarrow\tau$ and
  $B_2=\bigwedge_{i=k+1}^n{D_i}\Rightarrow\tau$. Remark that $C_1 \equiv
  A\Rightarrow B_1$, $C_2\equiv A\Rightarrow B_2$ and $B\equiv B_1\wedge B_2$.
\end{proof}

\begin{lemma}
  \label{lem:eqConjN}
  If $A\wedge B\equiv C\wedge D$ then one of the following cases happens
  \begin{enumerate}
  \item $A\equiv C_1\wedge D_1$ and $B\equiv C_2\wedge D_2$, with $C\equiv
    C_1\wedge C_2$ and $D\equiv D_1\wedge D_2$.
  \item $B\equiv C\wedge D_2$, with $D\equiv A\wedge D_2$.
  \item $B\equiv C_2\wedge D$, with $C\equiv A\wedge C_2$.
  \item $A\equiv C\wedge D_1$, with $D\equiv D_1\wedge B$.
  \item $A\equiv C_1\wedge D$, with $C\equiv C_1\wedge B$.
  \item $A\equiv C$ and $B\equiv D$.
  \item $A\equiv D$ and $B\equiv C$.
  \end{enumerate}
\end{lemma}
\begin{proof}
  Let $\PF A=R$, $\PF B=S$, $\PF C=T$, and $\PF D=U$. By Lemma~\ref{lem:eqPF},
  we have $R\uplus S\sim T\uplus U$. We prove first that there exist four
  multisets $V$, $W$, $X$, and $Y$ such that $R = V\uplus X$, $S = W\uplus Y$, $T
  = V\uplus W$, and $U = X\uplus Y$. Notice that $V$ and $X$ cannot be both empty,
  $W$ and $Y$ cannot be both empty, $V$ and $W$ cannot be both empty, and $X$ and
  $Y$ cannot be both empty.

  We have $T\uplus(S\cap U) = (T\uplus S)\cap (T\uplus U) \sim (T\uplus S)\cap
  (R\uplus S) = (T\cap R)\uplus S$. Thus, $T\setminus (T\cap R) \sim S\setminus
  (S\cap U)$. In the same way, $R\setminus (R\cap T) \sim U\setminus (S\cap U)$.
  We take $V=R\cap T$, $Y=S\cap U$, $W=T\setminus V\sim S\setminus Y$,
  $X=R\setminus V\sim U\setminus Y$.

  Now, if $V, W, X, Y$ are all non-empty, we let $C_1=\conj V$, $C_2=\conj W$,
  $D_1=\conj X$, and $D_2=\conj Y$, and we are in the first case.

  If $V$ is empty and the others are not, then we have $T=W$, $R=X$, so $A=\conj
  X$ and $C=\conj W$. We let $D_2=\conj Y$, hence we are in the second case.

  The cases where $W$, $X$, or $Y$ are empty, but the others are not, are
  symmetric.

  Finally, if $X$ and $W$ are both empty, then $A\equiv C$ and $B\equiv D$, and
  we are in the case 6. If $V$ and $Y$ are both empty, then $A\equiv D$ and
  $B\equiv C$, and we are in case 7.
\end{proof}

\begin{lemma}\label{lem:ImpImp}
  If $A\Rightarrow B\equiv C\Rightarrow\tau$, then either ($A\equiv C$ and
  $B\equiv\tau$), or ($C\equiv A\wedge B'$ and $B\equiv B'\Rightarrow\tau$).
\end{lemma}
\begin{proof}
  By Lemma~\ref{lem:eqPF}, $\PF{A\Rightarrow
    B}\sim\PF{C\Rightarrow\tau}=[C\Rightarrow\tau]$. Let $\PF
  B=[B_i\Rightarrow\tau]_{i=1}^n$. Then $\PF{A\Rightarrow B}=[(A\wedge
  B_i)\Rightarrow\tau]_{i=1}^n$. Therefore, $n=1$ and $A\wedge B_1\equiv C$. If
  $B_1=\varnothing$, then $A\equiv C$ and $B\equiv\tau$. If $B_1\neq\varnothing$,
  then $A\wedge B_1\equiv C$ and $B\equiv B_1\Rightarrow\tau$.
\end{proof}

\begin{lemma}\label{lem:WedgeWedgeSameLeft}
  If $A\wedge B\equiv A\wedge C$, then $B\equiv C$.
\end{lemma}
\begin{proof}
  By Lemma~\ref{lem:eqPF}, $\PF{A\wedge
    B}=\PF{A}\uplus\PF{B}\sim\PF{A}\uplus\PF{C}=\PF{A\wedge C}$. Then $\PF B\sim\PF
  C$, and so, by Lemma~\ref{lem:eqConjPF}, $B\equiv C$.
\end{proof}

\begin{lemma}\label{lem:ImpImpSamePremise}
  If $A\Rightarrow B\equiv A\Rightarrow C$, then $B\equiv C$.
\end{lemma}
\begin{proof}
  Let $\PF{A\Rightarrow B} = [(A\wedge B_i)\Rightarrow\tau]_{i=1}^n$, with
  $[B_i\Rightarrow\tau]_{i=1}^n=\PF B$, and $\PF{A\Rightarrow C} = [(A\wedge
  C_i)\Rightarrow\tau]_{i=1}^m$, with $[C_i\Rightarrow\tau]_{i=1}^n=\PF C$. By
  Lemma~\ref{lem:eqPF}, $n=m$ and, without lost of generality, we can consider
  that $(A\wedge B_i)\Rightarrow\tau\equiv (A\wedge C_i)\Rightarrow\tau$. Then, by
  Lemma~\ref{lem:ImpImp}, $A\wedge B_i\equiv A\wedge C_i$, so, by
  Lemma~\ref{lem:WedgeWedgeSameLeft}, $B_i\equiv C_i$. Therefore, by
  Lemma~\ref{lem:eqConjPF}, $B\equiv (B_1\Rightarrow\tau)\wedge\cdots\wedge
  (B_n\Rightarrow\tau)\equiv (C_1\Rightarrow\tau)\wedge\cdots\wedge
  (C_n\Rightarrow\tau)\equiv C$.
\end{proof}

\section{The \texorpdfstring{\OCe}{System I-eta}}\label{sec:calculus}
\subsection{Syntax}
We associate to each type $A$ (up to equivalence) an infinite set of variables
$\V_A$ such that if $A\equiv B$ then $\V_A=\V_B$ and if $A\not\equiv B$ then
$\V_A\cap\V_B=\varnothing$. The set of preterms is defined by
\[
   r\ =\ x~|~\lambda x.r~|~ rr~|~r\times r~|~\pi_A(r)
\]
These terms are called respectively, variables, abstractions, applications,
products and projections. An introduction is either an abstraction or a product.
An elimination is either an application or a projection. We recall the type on
binding occurrences of variables and write $\lambda x^A.t$ for $\lambda x.t$
when $x\in\V_A$. The set of free variables of $r$ is written $\FV(r)$.
$\alpha$-equivalence and substitution are defined as usual. The type system is
given in Table~\ref{tab:typeSys}. We use a presentation of typing rules without
explicit context
following~\cite{GeuversKrebbersMcKinnaWiedijkLFMTP10,ParkSeoParkLeeJAR13}, hence
the typing judgments have the form $ r:A$. The well-typed preterms are called
terms.

\begin{table}[t]
  \[
    \begin{array}{c@{\qquad}c@{\qquad}c}
      \condi{\cond{x\in\V_A}}{\infer[^{(ax)}]{x:A}{\phantom{x:A}}}
      &
        \condi{\cond{A\equiv B}}{\infer[^{(\equiv)}]{ r:B}{ r:A}}
      &
        \infer[^{(\Rightarrow_i)}]{\lambda x^A. r:A\Rightarrow B}{ r:B}
      \\[1em]
      {\infer[^{(\Rightarrow_e)}]{ r s:B}{ r:A\Rightarrow B &  s:A}}
      &
        \infer[^{(\wedge_i)}]{ r\times s:A\wedge B}{ r:A &  s:B}
      &
        \infer[^{(\wedge_e)}]{\pi_A( r):A}{ r:A\wedge B}
    \end{array}
  \]
  \caption{The type system.}
  \label{tab:typeSys}
\end{table}

\subsection{Operational semantics}\label{sec:opSem}
The operational semantics of the calculus is defined by two relations: an
equivalence relation, and a reduction relation.

\begin{definition}
  The symmetric relation $\eq$ is the smallest contextually closed relation
  defined by the rules given in Table~\ref{tab:opSemSym}.
\end{definition}

\begin{table}[t]
  \begin{align*}
    r\times  s &\eq  s\times  r & \rulelabel{(comm)} \\
    ( r\times  s)\times  t &\eq  r\times ( s\times  t) & \rulelabel{(asso)}\\
    \lambda x^A.( r\times  s) &\eq \lambda x^A. r\times \lambda x^A. s & \rulelabel{(dist)} \\
    r s t &\eq  r( s\times  t) & \rulelabel{(curry)}
  \end{align*}
  \caption{Symmetric relation.}
  \label{tab:opSemSym}
\end{table}

Because of the associativity property of $\times$, the term $ r\times (\ve
s\times t)$ is equivalent to the term $( r\times s)\times t$, so we can just
write it $ r\times s\times t$.

The size of a term $S(r)$, defined, as usual, by $S(x)=1$, $S(\lambda x^A.r) =
S(\pi_A(r))=1+S(r)$, $S(rs) = S(r\times s)=1+S(r)+S(s)$, is not invariant
through the equivalence $\eq$. Hence, we introduce a measure $M(\cdot)$ (given
in Table~\ref{tab:PandM}) which relies on a measure $P(\cdot)$ counting the
number of pairs in a term.

\begin{table}
  \[
    \begin{array}{rl|rl}
      P(x) &= 0& 
                 M(x) &= 1\\ 
      P(\lambda x^A.r) &= P(r)& 
                                M(\lambda x^A.r) &= 1 + M(r) + P(r)\\ 
      P(rs) &= 0& 
                  M(rs) &= 1 + M(r) + M(s)\\ 
      P(r \times s) &= 1 + P(r) + P(s)& 
                                        M(r \times s) &= 1 + M(r) + M(s)\\ 
      P(\pi_A(r)) &= 0 &
                         M(\pi_A(r)) &= 1 +  M(r)
    \end{array}
  \]
  \caption{Measure on terms.}
  \label{tab:PandM}
\end{table}

\begin{lemma}
  \label{lem:eqM}
  If $ r \eq s$ then
  $M(r) = M(s)$.
\end{lemma}
\begin{proof}
  First, we check the case of each rule of Table~\ref{tab:opSemSym}, and then
  conclude by structural induction to handle the contextual closure.
  \begin{itemize}
  \item \rulelabel{(comm)}:
    \( M( r \times s) =
    1 + M( r) + M( s) = M( s \times r). \)

  \item \rulelabel{(asso)}:
    \( M(( r \times s) \times t) = 2 + M( r) + M( s) + M( t) =
    M( r \times ( s \times t)). \)

  \item \rulelabel{(dist)}:
    \(
    M(\lambda x^A . ( r \times s))
    = 3 + M( r) + M( s) + P( r) + P( s)
    = M(\lambda x^A .  r \times \lambda x^A .  s)
    \)

  \item \rulelabel{(curry)}:
    \(
    M(r  s  t)
    = 2 + M( r)+M( s)+M( t)
    = M( r( s \times t))
    \)
    \qedhere
  \end{itemize}
\end{proof}

\begin{lemma}
  \label{lem:finiteClasses}
  For any term $ r$, the set $\{ s~|~ s \eq^* r\}$ is finite (modulo
  $\alpha$-equivalence).
\end{lemma}
\begin{proof}
  Let $F=\FV(r)$ and $n=M(r)$. We have $\{ s~|~ s \eq^* r\} \subseteq \{
  s~|~\FV( s) = F~\mbox{and}~M( s) = n\} \subseteq \{ s~|~\FV( s) \subseteq
  F~\mbox{and}~S(s) \leq n\}$. Hence, it is finite.
\end{proof}

\begin{definition}
  The reduction relation $\re$ is given in Table~\ref{tab:opSem}.
    This Table must be read in three steps: first we define the relation
    $\re_{\beta\pi\zeta}$, then the relation
    $\re_{\eta\delta}$, and finally the relations
    $\re$ and $\basicre$ in a mutually dependent way. Like in \cite{JayGhaniJFP95} this relation
    $\basicre$ forbids $\eta$-expansions and $\delta$-expansion and
    is used to reduce terms that are the left part of an application
    or the body of a projection.
\end{definition}

Since, in \OCe, an abstraction can be equivalent to a product, a subterm can
neither be $\eta$-expanded nor $\delta$-expanded, if it is either an abstraction
or a product, or if it occurs at left of an application or in the body of a
projection~\cite{DicosmoKesnerMSCS94}.

\begin{definition}
  We write $\toreq$ for the relation $\re$ modulo $\eq^*$ (i.e. $ r\toreq s$ iff
  $\ve r\eq^* r'\re s'\eq^* s$), and $\toreq^*$ for its transitive and reflexive
  closure.
  We write $t\basicr t'$ for the relation $\basicre$ modulo $\eq^*$ (i.e.
  $r\basicr s$ iff $r\eq^*r'\basicre s'\eq^*s$).
\end{definition}

\begin{remark}\label{rmk:finite}
  By Lemma~\ref{lem:finiteClasses}, a term has a finite number of one-step reducts
  and these reducts can be computed.
\end{remark}

Finally, notice that unlike in \OC, the $\zeta$-rule transforming an elimination
into an introduction is a reduction rule and not an equivalence rule. Hence,
variables, applications, and projections are preserved by $\eq$. In contrast, an
abstraction can be equivalent to a product, but, introductions are preserved.

\begin{table}[t]
  \begin{align*}
    \mbox{If } s:A,\  (\lambda x^A. r) s &\re_{\beta\pi\zeta}  r[ s/x] & \rulelabel{($\beta$)}\\
    \mbox{If } r:A,\ \pi_A( r\times s) &\re_{\beta\pi\zeta}  r & \rulelabel{($\pi$)}\\
    (r\times s)t&\re_{\beta\pi\zeta} rt\times st &\rulelabel{($\zeta$)}\\
    \mbox{\parbox{5.7cm}{\centering
    If $r:A\Rightarrow B$, $x$ fresh,\\
    and $r$ is an elimination or a variable,
    }}\ r &\re_{\eta\delta}\lambda x^A.(rx) &\rulelabel{($\eta$)}\\
    \mbox{\parbox{5.7cm}{\centering
    If $r:A\wedge B$ \\
    and $r$ is an elimination or a variable,
    }}\ r &\re_{\eta\delta}\pi_A(r)\times\pi_B(r)  &\rulelabel{($\delta$)}
  \end{align*}
  \[
    \infer{r\basicre s}{r\re_{\beta\pi\zeta} s}
    \qquad
    \infer{r\re s}{r\re_{\eta\delta} s}
    \quad
    \infer{r\re s}{r\basicre s}
    \qquad
    \infer{\lambda x.r\basicre\lambda x.s}{r\re s}
    \qquad
    \infer{rt\basicre st}{r\basicre s}
  \]
  \[
    \infer{tr\basicre ts}{r\re s}
    \qquad
    \infer{r\times t \basicre s\times t}{r\re s}
    \qquad
    \infer{t\times r\basicre t\times s}{r\re s}
    \qquad
    \infer{\pi_A(r)\basicre\pi_A(s)}{r\basicre s}
  \]
	\caption{Reduction relation.}
  \label{tab:opSem}
\end{table}

\section{Subject Reduction}\label{sec:SR}
The set of types assigned to a term is preserved under $\eq$ and $\re$. Before
proving this property, we prove the unicity of types (Lemma~\ref{lem:unicity}),
the generation lemma (Lemma~\ref{lem:generation}), and the substitution lemma
(Lemma~\ref{lem:substitution}). 

\begin{lemma}[Unicity]\label{lem:unicity}
  If $ r:A$ and $ r:B$, then $A\equiv B$.
\end{lemma}
\begin{proof}~
  \begin{itemize}
  \item If the last rule of the derivation of $ r:A$ is $(\equiv)$, then we have
    a shorter derivation of $r:C$ with $C\equiv A$, and, by the i.h., $C\equiv B$,
    hence $A\equiv B$.
  \item If the last rule of the derivation of $ r:B$ is $(\equiv)$ we proceed in
    the same way.
  \item All the remaining cases are syntax directed.
    \qedhere
  \end{itemize}
\end{proof}

\begin{lemma}[Generation]
  \label{lem:generation}
  ~
  \begin{enumerate}
  \item\label{case:var} If $x\in\V_A$ and $x:B$, then $A\equiv B$.
  \item\label{case:lambda} If $\lambda x^A. r:B$, then $B\equiv A\Rightarrow C$
    and $ r:C$.
  \item If $ r s:B$, then $ r:A\Rightarrow B$ and $ s:A$.
  \item If $ r\times s:A$, then $A\equiv B\wedge C$ with $ r:B$ and $ s:C$.
  \item If $\pi_A( r):B$, then $A\equiv B$ and $ r:B\wedge C$.
  \end{enumerate}
\end{lemma}
\begin{proof}
  Each statement is proved by induction on the typing derivation. For the
  statement~\ref{case:var}, we have $x\in\V_A$ and $x:B$. The only way to type
  this term is either by the rule $(ax)$ or $(\equiv)$.
  \begin{itemize}
  \item In the first case, $A=B$, hence $A\equiv B$.
  \item In the second case, there exists $B'$ such that $x:B'$ has a shorter
    derivation, and $B\equiv B'$. By the i.h. $A\equiv B'\equiv B$.
  \end{itemize}
  For the statement~\ref{case:lambda}, we have $\lambda x^A.\ve r:B$. The only
  way to type this term is either by rule $(\Rightarrow_i)$, $(\equiv)$.
  \begin{itemize}
  \item In the first case, we have $B=A\Rightarrow C$ for some, $C$ and $ r:C$.
  \item In the second, there exists $B'$ such that $\lambda x^A.\ve r:B'$ has a
    shorter derivation, and $B\equiv B'$. By the i.h., $B'\equiv A\Rightarrow C$ and
    $ r:C$. Thus, $B\equiv B'\equiv A\Rightarrow C$.
  \end{itemize}
  The three other statements are similar.
\end{proof}

\begin{lemma}[Substitution]
  \label{lem:substitution}
  If $ r:A$, $ s:B$, and $x\in\V_B$, then $ r[ s/x]:A$.
\end{lemma}
\begin{proof}
  By structural induction on $ r$.
  \begin{itemize}
  \item Let $ r=x$. By Lemma~\ref{lem:generation}, $A\equiv B$, thus $ s:A$. We
    have $x[ s/x]= s$, so $x[ s/x]:A$.

  \item Let $ r=y$, with $y\neq x$. We have $y[ s/x]=y$, so $y[ s/x]:A$.

  \item Let $ r=\lambda y^C. r'$. By Lemma~\ref{lem:generation}, $A\equiv
    C\Rightarrow D$, with $ r':D$. By the i.h., $r'[ s/x]:D$, and so, by rule
    $(\Rightarrow_i)$, $\lambda y^C. r'[ s/x]:C\Rightarrow D$. Since $\lambda y^C.
    r'[ s/x]=(\lambda y^C. r')[ s/x]$, using rule $(\equiv)$, $(\lambda y^C. r')[
    s/x]:A$.

  \item Let $ r= r_1 r_2$. By Lemma~\ref{lem:generation}, $ r_1:C\Rightarrow A$
    and $ r_2:C$. By the i.h.~$ r_1[ s/x]:C\Rightarrow A$ and $ r_2[ s/x]:C$, and
    so, by rule $(\Rightarrow_e)$, $( r_1[ s/x])( r_2[ s/x]):A$. Since $(\ve r_1[
    s/x])( r_2[ s/x])=( r_1 r_2)[ s/x]$, we have $( r_1 r_2)[ s/x]:A$.

  \item Let $ r= r_1\times r_2$. By Lemma~\ref{lem:generation}, $ r_1:A_1$ and $
    r_2:A_2$, with $A\equiv A_1\wedge A_2$. by the i.h. $\ve r_1[ s/x]:A_1$ and $
    r_2[ s/x]:A_2$, and so, by rule $(\wedge_i)$, $( r_1[ s/x])\times ( r_2[\ve
    s/x]):A_1\wedge A_2$. Since $( r_1[ s/x])\times ( r_2[\ve s/x])=( r_1\times
    r_2)[ s/x]$, using rule $(\equiv)$, we have $( r_1\times r_2)[ s/x]:A$.

  \item Let $ r=\pi_A( r')$. By Lemma~\ref{lem:generation}, $ r':A\wedge C$.
    Hence, by the i.h., $ r'[ s/x]:A\wedge C$. Hence, by rule $\wedge_e$, $\pi_A(
    r'[ s/x]):A$. Since $\pi_A( r'[\ve s/x])=\pi_A( r')[ s/x]$, we have $\pi_A(
    r')[\ve s/x]:A$.
    \qedhere
  \end{itemize}
\end{proof}

\begin{theorem}[Subject reduction]\label{thm:SR}
  If $ r:A$ and $ r\re s$ or $ r\eq s$ then $ s:A$.
\end{theorem}
\begin{proof}
  By induction on the rewrite relation.
  \begin{itemize}
  \item \rulelabel{(comm)}: 
    If $ r\times s:A$, then by Lemma~\ref{lem:generation}, $A\equiv A_1\wedge
    A_2\equiv A_2\wedge A_1$, with $ r:A_1$ and $ s:A_2$. Then, $s\times r:A_2\wedge
    A_1\equiv A$.
  \item \rulelabel{(asso)}: 
    ~
    \begin{description}
    \item[$(^{\to})$] If $( r\times s)\times t:A$, then by
      Lemma~\ref{lem:generation}, $A\equiv (A_1\wedge A_2)\wedge A_3\equiv
      A_1\wedge(A_2\wedge A_3)$, with $ r:A_1$, $\ve s:A_2$ and $ t:A_3$. Then, $
      r\times ( s\times t):A_1\wedge(A_2\wedge A_3)\equiv A$.
    \item[$(_{\leftarrow})$] Analogous to $(^{\to})$.
    \end{description}
  \item \rulelabel{(dist)}:~
    \begin{description}
    \item[$(^{\to})$] If $\lambda x^B.( r\times s):A$, then by
      Lemma~\ref{lem:generation}, we have $A\equiv (B\Rightarrow(C_1\wedge C_2))\equiv
      ((B\Rightarrow C_1)\wedge(B\Rightarrow C_2))$, with $ r:C_1$ and $ s:C_2$. Then,
      ${\lambda x^B. r\times \lambda x^B. s:(B\Rightarrow C_1)\wedge(B\Rightarrow
        C_2)}\equiv A$.
    \item[$(_{\leftarrow})$] If $\lambda x^B. r\times \lambda x^B. s:A$, then by
      Lemma~\ref{lem:generation}, $A\equiv((B\Rightarrow C_1)\wedge(B\Rightarrow
      C_2))\equiv (B\Rightarrow(C_1\wedge C_2))$, with $ r:C_1$ and $ s:C_2$. Then,
      ${\lambda x^B.( r\times \ve s):B\Rightarrow(C_1\wedge C_2)}\equiv A$.
    \end{description}
  \item \rulelabel{(curry)}:~ 
    \begin{description}
    \item[$(^{\to})$] If $ r s t:A$, then by Lemma~\ref{lem:generation}, $
      r:B\Rightarrow C\Rightarrow A\equiv(B\wedge C)\Rightarrow A$, $ s:B$ and $ t:C$.
      Then, $r(s\times t):A$.
    \item[$(_{\leftarrow})$] If $ r( s\times t):A$, then by
      Lemma~\ref{lem:generation}, $ r:(B\wedge C)\Rightarrow A\equiv (B\Rightarrow
      C\Rightarrow A)$, $ s:B$ and $ t:C$. Then $rst:A$.
    \end{description}

  \item \rulelabel{($\beta$)}: 
    If $(\lambda x^B. r) s:A$, then by Lemma~\ref{lem:generation}, $\lambda x^B.
    r:B\Rightarrow A$, and by Lemma~\ref{lem:generation} again, $ r:A$. Then by
    Lemma~\ref{lem:substitution}, $ r[ s/x^B]:A$.

  \item \rulelabel{($\pi$)}: 
    If $\pi_B( r\times s):A$, then by Lemma~\ref{lem:generation}, $A\equiv B$,
    and so, by rule $(\equiv)$, $ r:A$.

  \item \rulelabel{($\zeta$)}:~ 
    If $( r\times s) t:A$, then by Lemma~\ref{lem:generation}, $ r\times
    s:B\Rightarrow A$, and $ t:B$. Hence, by Lemma~\ref{lem:generation} again,
    $B\Rightarrow A\equiv C_1\wedge C_2$, and so by Lemma~\ref{lem:ImpConj},
    $A\equiv A_1\wedge A_2$, with $ r:B\Rightarrow A_1$ and $ s:B\Rightarrow A_2$.
    Then, $ r t\times s t:A_1\wedge A_2\equiv A$.

  \item\rulelabel{$(\eta)$}: If $r:A\Rightarrow B$, then, by rules
    $(\Rightarrow_e)$ and $(\Rightarrow_i)$, $\lambda x^A.(rx):A\Rightarrow B$.
  \item\rulelabel{$(\delta)$}: If $r:A\wedge B$, then by rules $(\wedge_e)$ and
    $(\wedge_i)$, $\pi_A(r)\times\pi_B(r):A\wedge B$.
    
  \item Contextual closure: Let $ t\to r$, where $\to$ is either $\eq$ or $\re$.
    \begin{itemize}
    \item Let $\lambda x^B. t\to\lambda x^B. r$: If $\lambda x^B. t:A$, then by
      Lemma~\ref{lem:generation}, $A\equiv (B\Rightarrow C)$ and $ t:C$, hence by the
      i.h., $ r:C$ and so $\lambda x^B.\ve r:B\Rightarrow C\equiv A$.
    \item Let $ t s\to r s$: If $ t s:A$ then by Lemma~\ref{lem:generation}, $
      t:B\Rightarrow A$ and $ s:B$, hence by the i.h., $\ve r:B\Rightarrow A$ and so $
      r s:A$.
    \item Let $ s t\to s t$: If $ s t:A$ then by Lemma~\ref{lem:generation}, $
      s:B\Rightarrow A$ and $ t:B$, hence by the i.h. $\ve r:B$ and so $ s r:A$.
    \item Let $ t\times s\to r\times s$: If $ t\times s:A$ then by
      Lemma~\ref{lem:generation}, $A\equiv A_1\wedge A_2$, $ t:A_1$, and $ s:A_2$,
      hence by the i.h., $ r:A_1$ and so $\ve r\times s:A_1\wedge A_2\equiv A$.
    \item Let $ s\times t\to s\times r$: Analogous to previous case.
    \item Let $\pi_B( t)\to\pi_B( r)$: If $\pi_B( t):A$ then by
      Lemma~\ref{lem:generation}, $A\equiv B$ and $ t:B\wedge C$, hence by the i.h. $
      r:B\wedge C$. Therefore, $\pi_B( r):B\equiv A$.
      \qedhere
    \end{itemize}
  \end{itemize}
\end{proof}

\section{Strong Normalization}\label{sec:SN}
We now prove the strong normalization of reduction $\toreq$. \medskip

\noindent\textbf{Road-map of the proof.} We associate, as usual, a set $\interp
A$ of strongly normalizing terms to each type $A$. We then prove an adequacy
lemma stating that every term of type $A$ is in $\interp A$. Compared with the
proof for simply typed lambda-calculus with pairs our proof presents several
novelties.
\begin{itemize}
\item In simply typed lambda-calculus, proving that if $r_1$ and $r_2$ are
  strongly normalizing, then so is $r_1\times r_2$ is easy. However, like in \OC,
  in \OCe this property is harder to prove, as it requires a characterization of
  the terms equivalent to the product $r_1\times r_2$ and of all its reducts. This
  will be the first part of our proof (Lemmas~\ref{lem:eqProd}, \ref{lem:reProd}
  and Corollary~\ref{cor:prodOfSN}).

\item As usual we associate to each type $A$ a
  set $\interp{A}$ of {\it reducible} terms,
  but this definition has to take into account the
  equivalence between types. For instance,
  $r\in\interp{\tau\Rightarrow(\tau\wedge \tau)}$, if and only if,
  $r:\tau\Rightarrow(\tau\wedge\tau)$, for all $s\in\interp{\tau}$,
  $rs\in\interp{\tau\wedge\tau}$, and, moreover,
  $\pi_{\tau\Rightarrow\tau}(r)\in\interp{\tau\Rightarrow\tau}$ as
  $\tau\Rightarrow(\tau\wedge\tau)\equiv(\tau\Rightarrow\tau)\wedge(\tau\Rightarrow\tau)$
  (Definition~\ref{def:reducibility}).

\item In the strong normalization proof of simply typed lambda-calculus the
  so-called properties CR1, CR2, and CR3, the adequacy of product, and the
  adequacy of abstraction are five independent lemmas. Like in
  \cite{JayGhaniJFP95}, we have to prove these properties in a huge single
  induction (Lemma~\ref{lem:tout}).
\item In simply typed lambda-calculus, neutral terms are those which are neither
  abstractions nor pairs. The reason is that such terms can be put in any context
  without creating a redex. In our case, the applications are not always neutral.
  For example, if $r : A$, $(\lambda x^{A\wedge B} .x)r$ is not neutral. Indeed,
  if $s:B$, $(\lambda x^{A\wedge B}.x)rs\eq(\lambda x^{A\wedge B}.x)(r\times s)\re
  r\times s$. This leads us to generalize the induction hypothesis in the proof of
  the adequacy of product and of abstraction.
\end{itemize}

The set of strongly normalizing terms is written $\SN$. The size of the longest
reduction issued from $t\in\SN$ is written $|t|$. Recall that each term has a
finite number of one-step reducts (Remark~\ref{rmk:finite}).
\begin{lemma} \label{lem:eqProd} If $ r\times  s\eq^* t$ then either
  \begin{enumerate}
  \item\label{it:eqProd-sum} $ t= u\times  v$ where either
    \begin{enumerate}
    \item $ u\eq^* t_{11}\times t_{21}$ and $ v\eq^*\ve t_{12}\times t_{22}$
      with $ r\eq^* t_{11}\times t_{12}$ and $ s\eq^* t_{21}\times t_{22}$, or
    \item $ v\eq^* w\times s$ with $ r\eq^* u\times w$, or any of the three
      symmetric cases, or
    \item $ r\eq^* u$ and $ s\eq^* v$, or the symmetric case.
    \end{enumerate}
  \item\label{it:eqProd-lam} $ t=\lambda x^A. a$ and $ a\eq^* a_1\times \ve a_2$
    with $ r\eq^*\lambda x^A. a_1$ and $ s\eq^*\lambda x^A.\ve a_2$.
  \end{enumerate}
\end{lemma}
\begin{proof}
  By a double induction, first on $M( t)$ and then on the length of the
  derivation of $r\times s\eq^* t$. Consider an equivalence proof $ r \times s
  \eq^* t' \eq\ve t$ with a shorter proof $ r \times s \eq^* t'$. By the second
  i.h. (induction hypothesis), the term $ t'$ has the form prescribed by the
  lemma. We consider the three cases and in each case, the possible rules
  transforming $ t'$ in $ t$.
  \begin{enumerate}
  \item Let $ r\times  s\eq^* u\times  v\eq t$. The possible
    equivalences from $ u\times  v$ are
    \begin{itemize}
    \item $ t= u'\times  v$ or $ u\times  v'$ with $ u\eq\ve
      u'$ and $ v\eq v'$, and so the term $ t$ is in case \ref{it:eqProd-sum}.
    \item Rules \rulelabel{(comm)} and \rulelabel{(asso)} preserve the conditions of case
      \ref{it:eqProd-sum}.
    \item $ t=\lambda x^A.( u'\times  v')$, with $ u=\lambda x^A.\ve
      u'$ and $ v=\lambda x^A. v'$. By the first i.h.
      (since $M( u)<M( t)$ and $M( v)<M( t)$), either
      \begin{enumerate}
      \item $ u\eq^* w_{11}\times  w_{21}$ and $ v\eq^*\ve
        w_{12}\times  w_{22}$, by the first i.h., $\ve
        w_{ij}\eq^*\lambda x^A. t_{ij}$ for $i=1,2$ and $j=1,2$, with $\ve
        u'\eq^* t_{11}\times  t_{21}$ and $ v'\eq^* t_{12}\times \ve
        t_{22}$, so $ u'\times  v'\eq^* t_{11}\times  t_{12}\times
        t_{21}\times  t_{22}$. Hence, $ r\eq^*\lambda x^A.(\ve
        t_{11}\times  t_{12})$ and $ s\eq^*\lambda x^A.( t_{21}\times
        t_{22})$, and hence the term $t$ is in case \ref{it:eqProd-lam}.
      \item $ v\eq^* w\times  s$ and $ r\eq^* u\times  w$.
        Since $ v\eq^*\lambda x^A. v'$, by the first i.h.,
        $ w\eq^*\lambda x^A. t_1$ and $ s\eq^*\lambda x^A. t_2$,
        with $ v'\eq^* t_1\times  t_2$. Hence, $ r\eq^*\lambda
        x.( u'\times  t_1)$, and hence the term $t$ is in case \ref{it:eqProd-lam}.
      \item $ r\eq^*\lambda x^A. u'$ and $ s\eq^*\lambda x^A. v$,
        and hence the term $t$ is in case \ref{it:eqProd-lam}.
      \end{enumerate} (the symmetric cases are analogous).
    \end{itemize}
  \item Let $ r\times  s\eq^*\lambda x^A. a\eq t$, with $\ve
    a\eq^* a_1\times  a_2$, $ r\eq^*\lambda x^A. a_1$, and $\ve
    s\eq^*\lambda x^A. a_2$. Hence, possible equivalences from $\lambda
    x. a$ to $ t$ are
    \begin{itemize}
    \item $ t=\lambda x^A. a'$ with $ a\eq^* a'$, hence $\ve
      a'\eq^* a_1\times  a_2$, and so the term $t$ is in case \ref{it:eqProd-lam}.
    \item $ t=\lambda x^A. u\times \lambda x^A. v$, with $ a_1\times
       a_2\eq^* a= u\times  v$. Hence, by the first i.h.
      (since $M( a)<M( t)$), either
      \begin{enumerate}
      \item $ a_1\eq^* u$ and $ a_2\eq^* v$, and so $\ve
        r\eq^*\lambda x^A. u$ and $ s\eq^*\lambda x^A. v$, or
      \item $ v\eq^* t_1\times  t_2$ with $ a_1\eq^* u\times \ve
        t_1$ and $ a_2\eq^* t_2$, and so $\lambda x^A. v\eq^*\lambda
        x. t_1\times \lambda x^A. t_2$, $ r\eq^*\lambda x^A.\ve
        u\times \lambda x^A. t_1$ and $ s\eq^*\lambda x^A. t_2$, or
      \item $ u\eq^* t_{11}\times  t_{21}$ and $ v\eq^*\ve
        t_{12}\times  t_{22}$ with $ a_1\eq^* t_{11}\times  t_{12}$
        and $ a_2\eq^* t_{21}\times  t_{22}$, and so $\lambda x^A.\ve
        u\eq^*\lambda x^A. t_{11}\times \lambda x^A. t_{21}$, $\lambda
        x. v\eq^*\lambda x^A. t_{12}\times \lambda x^A. t_{22}$, $\ve
        r\eq^*\lambda x^A. t_{11}\times \lambda x^A. t_{12}$ and $\ve
        s\eq^*\lambda x^A. t_{21}\times \lambda x^A. t_{22}$.
      \end{enumerate} (the symmetric cases are analogous), and so the term $t$
      is in case \ref{it:eqProd-sum}.
      \qedhere
    \end{itemize}
  \end{enumerate}
\end{proof}

\begin{lemma}
  \label{lem:reProd}
  If $ r_1\times r_2\eq^* s\re t$, there exists $ u_1$, $\ve u_2$ such that $
  t\eq^* u_1\times u_2$ and either ($ r_1\toreq u_1$ and $\ve r_2\eq^* u_2$), or
  ($ r_1\eq^* u_1$ and $ r_2\toreq u_2$).
\end{lemma}
\begin{proof}
  By induction on $M( r_1\times r_2)$. By Lemma~\ref{lem:eqProd}, $ s$ is either
  a product $s_1\times s_2$ or an abstraction $\lambda x^A.a$ with the conditions
  given in the lemma. The different terms $ s$ reducible by $\re$ are $ s_1\times
  s_2$ or $\lambda x^A. a$, with a reduction in the subterm $ s_1$, $ s_2$, or $
  a$.
  
  Notice that no rule can be applied in head position. Indeed, rule nor
  \rulelabel{($\beta$)} nor \rulelabel{($\zeta$)} can apply, since $s$ is not an
  application, rule \rulelabel{($\pi$)} cannot apply since $s$ is not a
  projection, and rules \rulelabel{($\eta$)} and \rulelabel{($\delta$)} cannot
  apply since $s$ is an introduction.

  We consider each case:
  \begin{itemize}
  \item $ s= s_1\times s_2$, $ t= t_1\times s_2$ or $\ve t= s_1\times t_2$, with
    $ s_1\re t_1$ and $ s_2\re t_2$. We only consider the first case since the other
    is analogous. One of the following cases happen
    \begin{enumerate}
    \item[(a)] $ r_1\eq^* w_{11}\times w_{21}$, $ r_2\eq^*\ve w_{12}\times
      w_{22}$, $ s_1= w_{11}\times w_{12}$ and $\ve s_2= w_{21}\times w_{22}$. Hence,
      by the i.h., either $ t_1= w'_{11}\times w_{12}$ or $ t_1=\ve w_{11}\times
      w'_{12}$, with $ w_{11}\re w_{11}'$ and $ w_{12}\re w_{12}'$. We take, in the
      first case $ u_1= w_{11}'\times w_{21}$ and $ u_2=\ve w_{12}\times
      w_{22}\eq^*r_2$, in the second case $ u_1= w_{11}\times \ve
      w_{21}\eq^*r_1$ and $ u_2= w_{12}'\times w_{22}$.
    \item[(b)] We consider two cases, since the other two are symmetric.
      \begin{itemize}
      \item $ r_1\eq^* s_1\times w$ and $ s_2\eq^* w\times \ve r_2$, in which
        case we take $ u_1= t_1\times w$ and $\ve u_2= r_2$.
      \item $ r_2\eq^* w\times s_2$ and $ s_1= r_1\times w$. Hence, by the i.h.,
        either $ t_1= r'_1\times \ve w$, or $ t_1= r_1\times w'$, with $ r_1\re r_1'$
        and $ w\re w'$. We take, in the first case $ u_1= r'_1$ and $ u_2= w\times s_2$,
        and in the second case $ u_1= r_1$ and $ u_2= w'\times s_2$.
      \end{itemize}
    \item[(c)] $ r_1\eq^* s_1$ and $ r_2\eq^* s_2$, in which case we take $ u_1=
      t_1$ and $ u_2= s_2$.

    \end{enumerate}
  \item $ s=\lambda x^A. s'$, $ t=\lambda x^A. t'$, and $ s'\re\ve t'$, with $
    s'\eq^* s'_1\times s'_2$ and $ s\eq^*\lambda x^A.\ve s'_1\times \lambda x^A.
    s'_2$. Therefore, by the i.h., then there exists $ u'_1$, $ u'_2$ such that
    either ($ s'_1\eq^* u'_1$ and $s'_2\toreq u'_2$) or ($s'_1\toreq u'_1$ and
    $s'_2\eq^* u'_2$). Therefore, we take $u_1=\lambda x^A. u_1'$ and $u_2=\lambda
    x^A. u_2'$.
    \qedhere
  \end{itemize}
\end{proof}

\begin{corollary}\label{cor:prodOfSN}
  If $ r_1\in\SN$ and $ r_2\in\SN$, then $ r_1\times  r_2\in\SN$.
\end{corollary}
\begin{proof}
  By Lemma~\ref{lem:reProd}, from a reduction sequence starting from $\ve
  r_1\times r_2$, we can extract one starting from $ r_1$, $r_2$, or both. Hence,
  this reduction sequence is finite.
\end{proof}

\begin{lemma}
  \label{lem:lamOfSN}
  If $r\in\SN$, then $\lambda x^A.r\in\SN$.
\end{lemma}
\begin{proof}
  By induction on the length of the derivation we prove that if $\lambda
  x^A.r\eq^* s$, then $s=(\lambda x^A.s_1)\times\cdots\times(\lambda x^A.s_n)$,
  where $r\eq^*s_1\times\cdots\times s_n$. Thus, if $\lambda x^A.r\eq^*s\re t$,
  the reduction is in some $s_i$, thus $t\eq^*\lambda x^A.r'$ where $r\toreq r'$.
  Therefore, $\lambda x^A.r\in\SN$.
\end{proof}

\begin{lemma}
  \label{lem:IntroAppEqIntroApp}
  Let $r$ and $t$ be introductions, then if $rs\eq^*tu$, then $r\eq^*t$ and
  $s\eq^*u$.
\end{lemma}
\begin{proof}
  We proceed by induction on the length of the derivation $rs\eq v\eq^*tu$. So,
  the possibilities for $v$ are:
  \begin{enumerate}
  \item If $v=r's$ or $v=rs'$, with $r\eq r'$ and $s\eq s'$, the i.h.~applies.
  \item If $v$ is obtained by \rulelabel{(curry)}, then either $r=r_1r_2$, which
    is impossible since no elimination is equivalent to an introduction, or
    $s=s_1\times s_2$, and $v=rs_1s_2$, then by the i.h., we have $rs_1\eq^*t$,
    which is impossible since no elimination is equivalent to an introduction.
    \qedhere
  \end{enumerate}
\end{proof}

\begin{definition}[Reducibility]\label{def:reducibility}
  The set $\interp A$ of reducible terms of type $A$ is defined by induction on
  $m(A)$ as follows: $t\in\interp A$ if and only if $t:A$ and
  \begin{itemize}
  \item if $A\equiv\tau$, then $t\in\SN$,
  \item for all $B$, $C$, if $A\equiv B\Rightarrow C$, then for all
    $r\in\interp{B}$, $tr\in\interp C$,
  \item for all $B$, $C$, if $A\equiv B\wedge C$, then $\pi_B(t)\in\interp B$.
  \end{itemize}
\end{definition}
Note that, by construction, if $A\equiv B$, then $\interp A=\interp B$.

\begin{definition}[Neutral term]
  A term $t$ is neutral if no term of the form $tr$ or $\pi_A(t)$, can be
  $\basicr$-reduced at head position.
\end{definition}

The variables and the projections are always neutral, but, as we have discussed in
the road-map of the proof, applications are not necessarily neutral. For
example if $r : A$, then $(\lambda x^{A\wedge B} .x)r$ is not.

\begin{lemma}\label{lem:tout}
  For all types $T$, we have
  \begin{itemize}
  \item \titre{CR1} $\interp T\subseteq\SN$.
  \item \titre{CR2} If $t\in\interp T$ and $t\toreq t'$, then $t'\in\interp T$.
  \item \titre{CR3'} If $t:T$ is neutral, and for all $t'$ such that $t\basicr
    t'$, $t'\in\interp T$, we have $t\in\interp T$.
  \item \titre{Adequacy of product} If $T=A\wedge B$, then for all $r\in\interp
    A$ and $s\in\interp B$, $r\times s\in\interp{T}$.
  \item \titre{Adequacy of abstraction} If $T=A\Rightarrow B$, then for all
    $t\in\interp B$, if for all $r\in\interp A$, $t[r/x]\in\interp B$, then $\lambda
    x^A.t\in\interp{T}$.
  \end{itemize}
\end{lemma}
\begin{proof}
  By induction on $m(T)$.
  
  \proofof{\titre{CR1}}
  Let $t\in\interp T$. We want to prove that $t\in\SN$.
  \begin{itemize}
  \item If $T=\tau$, then $t\in\interp T=\SN$.
  \item If $T=A\Rightarrow B$,
    then, by the i.h. \titre{CR3'}, we have $x^A\in\interp A$. Hence,
    $tx\in\interp B$, then, by the i.h.,
    $tx\in\SN$.
    We prove by a second induction on $|tx|$ that all the one-step $\toreq$-reducts of
    $t$ are in $\SN$.
    \begin{itemize}
    \item If $t\basicr t'$, then $tx\basicr t'x$, so by the second i.h., $t'\in\SN$.
    \item If $t\toreq_\eta \lambda y^C.(ty)$, where $T\equiv C\Rightarrow D$.
      Since $t\in\interp T$, and, by the i.h. \titre{CR3'}, $y\in\interp C$, so
      $ty\in\interp D$, which, by the i.h. is a subset of $\SN$. Therefore, by
      Lemma~\ref{lem:lamOfSN}, $\lambda y^C.(ty)\in\SN$.
    \item If $t\toreq_\delta\pi_C(t)\times\pi_D(t)$, where $T\equiv C\wedge
      D$. Since $t\in\interp T$, we have $\pi_C(t)\in\interp C$, and by the i.h.,
      $\pi_C(t)\in\SN$. In the same way, $\pi_D(t)\in\SN$, so by
      Corollary~\ref{cor:prodOfSN}, $\pi_C(t)\times\pi_D(t)\in\SN$.
    \end{itemize}
  \item If $T=A\wedge B$, then $\pi_A(t)\in\interp A$ and
    $\pi_B(t)\in\interp B$. by the i.h., $\interp
    A\subseteq\SN$, and so we proceed by a second induction on $|\pi_A(t)|$ to prove
    that all the one-step $\toreq$-reducts of $t$ are in $\SN$.
    \begin{itemize}
    \item If $t\basicr t'$, $\pi_A(t)\basicr\pi_A(t')$, so by the second i.h., $t'\in\SN$.
    \item If $t\toreq_\eta \lambda y^C.(ty)$, where $T\equiv C\Rightarrow D$.
      Since $t\in\interp T$, and, by the i.h.
      \titre{CR3'},
      $y\in\interp C$, so $ty\in\interp D$, which, by the i.h.
      is a subset of $\SN$.
      Therefore, by Lemma~\ref{lem:lamOfSN}, $\lambda y^C.(ty)\in\SN$.
    \item If $t\toreq_\delta\pi_C(t)\times\pi_D(t)$, where $T\equiv C\wedge D$. Since $t\in\interp T$, we have $\pi_C(t)\in\interp
      C$, and by the i.h.,
      $\pi_C(t)\in\SN$. In the same way, $\pi_D(t)\in\SN$, so by Corollary~\ref{cor:prodOfSN},
      $\pi_C(t)\times\pi_D(t)\in\SN$.
    \end{itemize}
  \end{itemize}
  
  \proofof{\titre{CR2}}
  Let $t\in\interp T$ and $t\toreq t'$. We want to prove that $t'\in\interp T$.
  Cases:
  \begin{itemize}
  \item $t\toreq_\triangle t'$. We want to prove that $t'\in\interp T$. That is, if $T\equiv\tau$, then $t'\in\SN$, if $T\equiv A\Rightarrow B$, then for all $r\in\interp A$, $t'r\in\interp B$, and if $T\equiv A\wedge B$, then $\pi_A(t')\in\interp A$.
    \begin{itemize}
    \item If $T\equiv\tau$, then since $t\in\SN$, we have $t'\in\SN$.
    \item If $T\equiv A\Rightarrow B$, then let $r\in\interp A$, we need to
      prove $t'r\in\interp B$. Since $t\in\interp T=\interp{A\Rightarrow B}$,
      we have $tr\in\interp B$. Then, by the i.h. in $\interp
      B$, and the fact that $tr\toreq_\triangle t'r$, we have $t'r\in\interp
      B$.
    \item If $T\equiv A\wedge B$, then we need to prove $\pi_A(t')\in\interp
      A$.
      Since $t\in\interp T=\interp{A\wedge B}$, we have $\pi_A(t)\in\interp
      A$. Then, by the i.h.~in $\interp A$, and the fact that $\pi_A(t)\toreq_\triangle \pi_A(t')$,
      we have $\pi_A(t')\in\interp A$.
    \end{itemize}
  \item $t\toreq_\eta \lambda x^A.tx$. Then, $T\equiv A\Rightarrow B$.
    Since $t\in\interp T=\interp{A\Rightarrow B}$, for any $s\in\interp A$, $ts\in\interp B$, and, since
    $x\notin \FV(t)$, we have $ts = (tx)[s/x]$. Then, by i.h.~\titre{Adequacy of abstraction}, $\lambda
    x^A.tx\in\interp{A\Rightarrow B}=\interp T$.
  \item $t\toreq_\delta\pi_A(t)\times\pi_B(t)$. Then, $T\equiv A\wedge B$. Since $t\in\interp T=\interp{A\wedge B}$, we have
    $\pi_A(t)\in\interp A$ and $\pi_B(t)\in\interp B$. Then, by the i.h.~\titre{Adequacy of product},
    $\pi_A(t)\times\pi_B(t)\in\interp{A\wedge B}=\interp T$.
  \end{itemize}

  \proofof{\titre{CR3'}}
  Let $t:T$ be a neutral term whose $\basicr$-one-step reducts $t'$
  are all in $\interp T$.
  We want to
  prove that $t\in\interp T$. That is, if $T\equiv\tau$, then $t\in\SN$, if $T\equiv A\Rightarrow B$, then for all $r\in\interp A$, $tr\in\interp B$, and if $T\equiv A\wedge B$, then $\pi_A(t)\in\interp A$.
  \begin{itemize}
  \item If $T\equiv\tau$, we need to prove that all the one-step reducts of
    $t$ are in $\SN$. 
    Since $T\equiv\tau$, these reducts are neither \rulelabel{($\eta$)}
    reducts nor
    \rulelabel{($\delta$)} reducts, but $\basicr$-reducts, which are in $\SN$. 
  \item If $T\equiv A\Rightarrow B$, we know that for all
    $r\in\interp A$, we have $t'r\in\interp B$. By
    the i.h.~\titre{CR1} in $\interp A$, we know $r\in\SN$. So we proceed by induction on
    $|r|$ to prove that $tr\in\interp B$. by the i.h., it
    suffices to check that every term $s$ such that
    $tr\basicr s$ is in $\interp B$.
    Since the reduction is $\basicr$, and the term $t$ is neutral, there is
    no possible head reduction. So, the possible cases are
    \begin{itemize}
    \item $s=tr'$ with $r\toreq r'$, then the i.h.~applies.
    \item $s=t'r$, with $t\toreq t'$. As $t$ cannot reduce to $t'$ by \rulelabel{($\delta$)} or
      \rulelabel{($\eta$)}, we have $t\basicr t'$, and $t'r\in\interp B$ by hypothesis.
    \end{itemize}
  \item If $T\equiv A\wedge B$, then we know that $\pi_A(t')\in\interp A$.
    by the i.h., it suffices to check that every term $s$ such that
    $\pi_A(t)\basicr s$ is in $\interp A$.
    Since the reduction is $\basicr$, and the term $t$ is neutral, there is
    no possible head reduction. So, the only possible case is
    $s=\pi_A(t')$ with $t\toreq t'$. As $t$ cannot reduce to $t'$ by \rulelabel{($\delta$)} or
    \rulelabel{($\eta$)}, we have $t\basicr t'$, and $\pi_A(t')\in\interp B$ by hypothesis.
  \end{itemize}

  \proofof{\titre{Adequacy of product}}
  If $T=A\wedge B$, we want to prove that for all $r\in\interp A$ and $s\in\interp
  B$, we have $r\times s\in\interp{T}$. We prove, more generally, by a simultaneous second induction on $m(D)$ that for all types
  $D$
  \begin{enumerate}
  \item if $T=A\wedge B\equiv D$, then $v=r\times s\in\interp{D}$, and
  \item if $T=A\wedge B\equiv C\Rightarrow D$, then for all $t\in\interp{C}$ we have
    $v=(r\times s)t\in\interp{D}$.
  \end{enumerate}

  To prove that $v\in \interp D$, we need to prove that if $D\equiv\tau$,
  then $v\in\SN$, if $D\equiv E\Rightarrow F$, then for all
  $u\in\interp E$, $vu\in\interp F$, and if $D\equiv E\wedge F$, then $\pi_{E}(v)\in\interp{E}$.

  \begin{itemize}
  \item $D\not\equiv\tau$, since, in case 1, it is equivalent to a conjunction,
    and also in case 2,
    by Lemma~\ref{lem:ImpConj}.
  \item If $D\equiv E\Rightarrow F$, in both cases we must prove that for all $u\in\interp
    E$, $vu\in\interp F$.
    \begin{enumerate}
    \item In case 1, we want to prove that $(r\times s)u\in\interp F$. Since
      $m(F)<m(D)$, the second i.h. applies.
    \item In case 2, we want to prove that $(r\times s)tu\in\interp F$.
      As $m(C\wedge E)<m((C\wedge E)\Rightarrow F)=m(T)$, by the i.h., $t\times u\in\interp{C\wedge E}$, and so, since
      $m(F)<m(D)$, by the second i.h., we have $(r\times s)(t\times
      u)\in\interp F$. Then, by the i.h. \titre{CR2}, $(r\times
      s)tu\in\interp F$.
    \end{enumerate}

  \item If $D\equiv E\wedge F$, in both cases we must prove that $\pi_E(v)\in\interp E$.
    \begin{itemize}
    \item In case 1, we want to prove that $\pi_E(r\times s)\in\interp E$.
      by the i.h. \titre{CR3'} it suffices to prove 
      that every one-step $\basicr$ reduct of $\pi_E(r\times s)$ is in $\interp E$.
      by the i.h. \titre{CR1}, $r,s\in\SN$, so we proceed with a third induction on $|r|+|s|$.

      A $\basicr$-reduction issued from $\pi_E(r\times s)$ cannot be a $\beta$-reduction or $\zeta$-reduction at head
      position, since a projection is not equivalent to an application (by rule inspection).
      Therefore, the possible $\basicr$-reductions issued from $\pi_E(r\times s)$ are:
      \begin{itemize}
      \item A reduction in $r\times s$, then, by
        Lemma~\ref{lem:reProd}, the reduction takes place either in $r$ or in $s$, and the
        third i.h.~applies.
      \item $\pi_E(r\times s)\eq^*\pi_E(w_1\times w_2)\re w_1$. Then, $r\times
        s\eq^*w_1\times w_2$. We need to prove that
        $w_1\in\interp E$.
        By Lemma~\ref{lem:eqProd}, we have either:
        \begin{itemize}
        \item $w_1\eq^*r_1\times s_1$, with $r\eq^*r_1\times r_2$ and $s\eq^*s_1\times
          s_2$. In such a case, by Lemma~\ref{lem:generation},
          $A\equiv A_1\wedge A_2$ and
          $B\equiv B_1\wedge B_2$,
          with
          $E\equiv A_1\wedge B_1$, and $F\equiv A_2\wedge B_2$. Since $r\in\interp A=\interp{A_1\wedge A_2}$, we have $\pi_{A_1}(r)\in\interp{A_1}$.
          Then, by the i.h. \titre{CR2} in $\interp{A_1}$, we have , $r_1\in\interp{A_1}$. 
          Similarly $s_1\in\interp{B_1}$. Then, by the i.h., 
          the i.h.~\titre{CR2}, $r_1\times s_1\eq^*w_1\in\interp{A_1\wedge B_1}=\interp E$.
          
        \item $w_1\eq^*r\times s_1$, with $s\eq^*s_1\times s_2$. Then, by
          Lemma~\ref{lem:generation}, $B\equiv B_1\wedge B_2$, with $E\equiv D_1$. Since $s\in\interp B=\interp{B_1\wedge B_2}$, we have $\pi_{B_1}(s)\in\interp{B_1}$.
          Then, by the i.h. \titre{CR2} in $\interp{B_1}$, we have $s_1\in\interp{B_1}$. Since,
          $r\in\interp A$, by the i.h. and the i.h.~\titre{CR2}, $r\times
          s_1\eq^*w_1\in\interp{D_1}=\interp E$.
          
        \item $w_1\eq^*r_1\times s$, with $r\eq^*r_1\times r_2$. This case is
          analogous to the previous one.

        \item $r\eq^*w_1\times r_2$, in which case, by
          Lemma~\ref{lem:generation}, $A\equiv E\wedge A_2$. since $r\in\interp A$,
          we have $\pi_E(r)\in\interp E$, so by the i.h. \titre{CR2} in $\interp E$,
          $w_1\in\interp E$.

        \item $s\eq^*w_1\times s_2$. This case is analogous to the
          previous case.

        \item $w_1\eq^*r\in\interp A=\interp E$.
        \item $w_1\eq^*s\in\interp B=\interp E$.
        \end{itemize}
      \end{itemize}
    \item In case 2, we want to prove that $\pi_E((r\times s)t)\in\interp E$.
      Since $T=A\wedge B\equiv C\Rightarrow D$, by Lemma~\ref{lem:ImpConj}, $D\equiv
      D_1\wedge D_2$, with $A\equiv C\Rightarrow D_1$ and $B\equiv C\Rightarrow D_2$.
      Since a projection is always neutral, and $m(E)<m(E\wedge F)=m(D)<m(C\Rightarrow
      D)=m(T)$, by i.h.~\titre{CR3'}, it suffices to prove that every one-step
      $\basicr$ reduction issued from $\pi_E((r\times s)t)$ is in $\interp E$. By the
      i.h. \titre{CR1}, $r,s,t\in\SN$. Therefore, we can proceed by a third induction
      on $|r|+|s|+|t|$. The reduction cannot happen at head position since a
      projection is not equivalent to an application, to apply $\beta$ or $\zeta$, and
      an application is not equivalent to a product to apply $\pi$. Hence, the
      reduction must happen in $(r\times s)t$. Therefore, we must prove that the
      one-step $\basicr$-reductions of $(r\times s)t$ are in $\interp
      D=\interp{E\wedge F}$, from which we conclude that $\pi_E((r\times
      s)t)\in\interp E$.
      
      A $\basicr$-reduction in $(r\times s)t$ cannot be a $\pi$-reduction in
      head position, since an application is not equivalent to a projection. Then, the
      possible $\basicr$ reductions issued from $(r\times s)t$ are:
      \begin{itemize}
      \item A reduction in $r\times s$, in which case, by Lemma~\ref{lem:reProd}
        it takes place either in $r$ or in $s$, and then the third i.h.~applies.
      \item A reduction in $t$, then the third i.h.~also applies.
      \item If the reduction is a $\beta$-reduction at head position, then we
        have $(r\times s)t\eq^*(\lambda x^C.w_1)w_2$. Hence, by
        Lemma~\ref{lem:IntroAppEqIntroApp}, $r\times s\eq^*\lambda x^A.w_1$ and
        $t\eq^*w_2$. By Lemma~\ref{lem:eqProd}, $r\eq^*\lambda x^C.r'$, $s\eq^*\lambda
        x^C.s'$, and $w_1\eq^*r'\times s'$. Therefore, $(r\times s)t\eq^*(\lambda
        x^C.r'\times s')t\re r'[t/x]\times s'[t/x]$. Since $(\lambda
        x^C.r')t\times(\lambda x^C.s')t \toreq^* r'[t/x]\times s'[t/x]$, by the i.h.
        \titre{CR2} in $\interp D$, it is enough to prove that $(\lambda
        x^C.r')t\times(\lambda x^C.s')t\in\interp D$. By the i.h. \titre{CR2}, since
        $r\in\interp A$ and $s\in\interp B$, we have, $r\eq^*\lambda x^C.r'\in\interp
        A=\interp{C\Rightarrow D_1}$, and $s\eq^*\lambda x^C.s'\in\interp
        B=\interp{C\Rightarrow D_2}$. Therefore, by definition, $(\lambda
        x^C.r')t\in\interp{D_1}$ and $(\lambda x^C.s')t\in\interp{D_2}$. Since
        $m(D)<m(T)$, by the i.h., we have $(\lambda x^C.r')t\times(\lambda
        x^C.s')t\in\interp{D}$.
      \item If the reduction is a $\zeta$-reduction at head position, then
        $(r\times s)t\eq^*(u_1\times u_2)w$. By Lemma~\ref{lem:IntroAppEqIntroApp},
        $r\times s\eq^*u_1\times u_2$ and $t\eq^*w$. By Lemma~\ref{lem:eqProd}, the
        possibilities are:
        \begin{itemize}
        \item $r\eq^*r_1\times r_2$, $s\eq^*s_1\times s_2$, $u_1\eq^*r_1\times
          s_1$ and $u_2\eq^*r_2\times s_2$. Then, $(u_1\times u_2)w\re_\zeta u_1w\times
          u_2w\eq^*(r_1\times s_1)w\times (r_2\times s_2)w$. By
          Lemmas~\ref{lem:generation} and~\ref{lem:ImpConj}, we have $D_1\equiv
          D_{11}\wedge D_{12}$ and $D_2\equiv D_{21}\wedge D_{22}$. So, since
          $r\in\interp{A}=\interp{C\Rightarrow D_1}=\interp{(C\Rightarrow
            D_{11})\wedge(C\Rightarrow D_{12})}$, we have $\pi_{C\Rightarrow
            D_{11}}(r)\in\interp{C\Rightarrow D_{11}}$, so, by the i.h. \titre{CR2},
          $r_1\in\interp{C\Rightarrow D_{11}}$. Similarly, $r_2\in\interp{C\Rightarrow
            D_{12}}$, $s_1\in\interp{C\Rightarrow D_{21}}$ and $s_2\in\interp{C\Rightarrow
            D_{22}}$. Therefore, by the i.h., $r_1\times s_1\in\interp{(C\Rightarrow
            D_{11})\wedge(C\Rightarrow D_{21})}=\interp{C\Rightarrow (D_{11}\wedge
            D_{21})}$, hence, by the i.h.~\titre{CR2}, we have
          $u_1\in\interp{C\Rightarrow(D_{11}\wedge D_{21})}$. Therefore,
          $u_1w\in\interp{D_{11}\wedge D_{21}}$. Similarly, $u_2w\in\interp{D_{12}\wedge
            D_{22}}$. So, by the i.h. again, $u_1w\times u_2w\in\interp{D_{11}\wedge
            D_{21}\wedge D_{12}\wedge D_{22}}=\interp D$.
        \item $s\eq^*s_1\times u_2$, $u_1\eq^*r\times s_1$. Then, $(u_1\times
          u_2)w\re_\zeta u_1w\times u_2w \eq^* (r\times s_1)w\times u_2w$. By
          Lemmas~\ref{lem:generation} and~\ref{lem:ImpConj}, we have $D_2\equiv
          D_{21}\wedge D_{22}$. So, since $s\in\interp{B}=\interp{C\Rightarrow
            D_2}=\interp{(C\Rightarrow D_{21})\wedge(C\Rightarrow D_{22})}$, we have
          $\pi_{C\Rightarrow D_{21}}(s)\in\interp{C\Rightarrow D_{21}}$, so, by the i.h.
          \titre{CR2}, $s_1\in\interp{C\Rightarrow D_{21}}$. Similarly,
          $u_2\in\interp{C\Rightarrow D_{22}}$. Therefore, by the i.h., we have that
          $r\times s_1\in\interp{(C\Rightarrow D_1)\wedge(C\Rightarrow
            D_{21})}=\interp{C\Rightarrow (D_1\wedge D_{21})}$, hence, by the
          i.h.~\titre{CR2}, $u_1\in\interp{C\Rightarrow(D_1\wedge D_{21})}$. Therefore,
          $u_1w\in\interp{D_1\wedge D_{21}}$. Similarly, $u_2w\in\interp{D_{22}}$. So, by
          the i.h.~again, $u_1w\times u_2w\in\interp{D_1\wedge D_{21}\wedge
            D_{22}}=\interp D$. The other three cases are symmetric.
        \item $r\eq^*u_1$ and $s\eq^*u_2$ or $r\eq^*u_2$ and $s\eq^*u_1$, then
          the $\zeta$-reduct of $(u_1\times u_2)w$ is $u_1w\times u_2w\eq^*rt\times st$.
          Hence, by the i.h. \titre{CR2} in $\interp{D_1}$, we have $rt\in\interp{D_1}$.
          Similarly, and $st\in\interp{D_2}$. Therefore, by the i.h., $rt\times
          st\in\interp{D_1\wedge D_2}=\interp D$.
        \end{itemize}
      \end{itemize}
    \end{itemize}
  \end{itemize}

  \proofof{\titre{Adequacy of abstraction}} If $T=A\Rightarrow B$, we want to
  prove that for all $t\in\interp B$, if for all $r\in\interp A$,
  $t[r/x]\in\interp B$, we have $\lambda x^A.t\in\interp{T}$. We prove, more
  generally, by a simultaneous second induction on $m(D)$ that for all type $D$
  \begin{enumerate}
  \item if $T=A\Rightarrow B\equiv D$, then $v=\lambda x^A.t\in\interp{D}$, and
  \item if $T=A\Rightarrow B\equiv C\Rightarrow D$, then for all
    $u\in\interp{C}$ we have $v=(\lambda x^A.t)u\in\interp{D}$.
  \end{enumerate}

  To prove that $v\in \interp D$, we need to prove that if $D\equiv\tau$, then
  $v\in\SN$, if $D\equiv E\Rightarrow F$, then for all $s\in\interp E$,
  $vs\in\interp F$, and if $D\equiv E\wedge F$, then $\pi_{E}(v)\in\interp{E}$.
  \begin{itemize}
  \item If $D\equiv\tau$, in both cases we must prove that $v\in\SN$.
    \begin{enumerate}
    \item  Case 1 is impossible, by Lemma~\ref{lem:generation}.
    \item In case 2, we have to prove that $v=(\lambda x^A.t)u\in\SN$, so it
      suffices to prove that every one-step $\basicr$ reduction issued from $(\lambda
      x^A.t)u$ is in $\SN$. by the i.h. \titre{CR1}, $t,u\in\SN$. Therefore, we can
      proceed by third induction on $|t|+|u|$. The possible $\basicr$ reductions
      issued from $(\lambda x^A.t)u$ are:
      \begin{itemize}
      \item Reducing $t$, or $u$, then the third i.h.~applies.
      \item $(\lambda x^A.t)u\toreq t[u/x]$, then, by
        Lemma~\ref{lem:generation}, we have $A\equiv C$, and by Lemma~\ref{lem:ImpImp},
        $B\equiv D$. Then, since by hypothesis $t[u/x]\in\interp B$, we have
        $t[u/x]\in\interp D=\SN$.
      \item $(\lambda x^A.t)u\toreq t[u_1/x]u_2$, with $u\eq^*u_1\times u_2$.
        Then, by Lemmas~\ref{lem:generation} and \ref{lem:ImpImp}, $C\equiv A\wedge C'$,
        and $C'\Rightarrow D\equiv B$ so, by definition of reducibility,
        $\pi_{A}(u)\in\interp{A}$ and $\pi_{C'}(u)\in\interp{C'}$. Therefore, by the
        i.h. \titre{CR2}, $u_1\in\interp{A}$ and $u_2\in\interp{C'}$.

        So, since $t[u_1/x]\in\interp{B}=\interp{C'\Rightarrow D}$, we have
        $t[u_1/x]u_2\in\interp D=\SN$.
      \item Notice that the reduction cannot be a $\zeta$-reduction in head
        position since, by $D\equiv\tau$ and so, by Lemma~\ref{lem:generation},
        $t\not\eq^*t_1\times t_2$ .
      \end{itemize}
    \end{enumerate}

  \item If $D\equiv E\Rightarrow F$, in both cases we must prove that for all $s\in\interp E$, we have
    $vs\in\interp F$.
    \begin{enumerate}
    \item In case 1, we have to prove that $(\lambda x^A.t)s\in\interp F$, which
      is a consequence of
      the second i.h., since $m(F)<m(D)$.

    \item In case 2, we have to prove that $(\lambda
      x^A.t)us\in\interp F$. Since $m(C\wedge E)<m((C\wedge E)\Rightarrow F)=m(T)$, by the i.h.~\titre{Adequacy of product}, $u\times s\in\interp{C\wedge E}$,
      then by the second i.h., since $m(F)<m(D)$, we have $(\lambda x^A.t)(u\times
      s)\in\interp F$, so, by the i.h. \titre{CR2}, $(\lambda
      x^A.t)us\in\interp F$.
    \end{enumerate}
  \item If $D\equiv E\wedge F$, in both cases we must prove that $\pi_{E}(v)\in\interp{E}$.
    \begin{enumerate}
    \item In case 1, we have to prove that $\pi_E(\lambda x^A.t)\in\interp E$.
      by the i.h. \titre{CR3'} it suffices to prove that every
      one-step $\basicr$ reduction issued from $\pi_{E}(\lambda x^A.t)$ is in
      $\interp{E}$. by the i.h. \titre{CR1}, $t\in\SN$. Therefore, we can proceed by third induction on
      $|t|$. The possible $\basicr$ reductions issued from $\pi_{E}(\lambda x^A.t)$
      are:
      \begin{itemize}
      \item A reduction in $t$, in which case, the third i.h.~applies.
      \item $\pi_E(\lambda x^A.t)\eq^*\pi_E(\lambda x^A.t_1\times\lambda
        x^A.t_2)\re\lambda x^A.t_1$. 
        By Lemmas~\ref{lem:generation} and \ref{lem:ImpConj}, $E\equiv A\Rightarrow
        E'$ and $F\equiv A\Rightarrow F'$, with $t_1:E'$ and $t_2:F'$. In addition, since
        $A\Rightarrow B\equiv T\equiv D\equiv E\wedge F\equiv A\Rightarrow (E'\wedge
        F')$, by Lemma~\ref{lem:ImpImpSamePremise}, we have
        $B\equiv E'\wedge F'$.
        Therefore, since $t[r/x]\in\interp B$, $\pi_{E'}(t[r/x])\in\interp{E'}$, by the
        i.h.~\titre{CR2}, $t_1[r/x]\in\interp{E'}$.
        We have $m(A\Rightarrow E')=m(E)<m(D)=m(T)=m(A\Rightarrow B)$, hence by the
        i.h., $\lambda x^A.t_1\in\interp{E}$.
      \end{itemize}
    \item In case 2, we have to prove that $\pi_E((\lambda x^A.t)u)\in\interp E$.
      by the i.h. \titre{CR3'} it suffices to prove that every
      one-step $\basicr$ reduction issued from $\pi_{E}((\lambda x^A.t)u)$ is in
      $\interp{E}$. by the i.h. \titre{CR1}, $t,u\in\SN$.
      Therefore, we can proceed by third induction on $|t|+|u|$. The possible
      $\basicr$ reductions issued from $\pi_{E}((\lambda x^A.t)u)$ are:
      \begin{itemize}
      \item A reduction in $t$ or in $u$, in which case, the third i.h.~applies.
      \item $\pi_{E}((\lambda x^A.t)u)\toreq\pi_{E}(t[u/x])$, hence by
        Lemmas~\ref{lem:generation} and \ref{lem:unicity}, $A\equiv C$,
        and so, by Lemma~\ref{lem:ImpImpSamePremise}, $B\equiv D\equiv E\wedge F$.
        Since $t[u/x]\in\interp B$, we have $\pi_{E}(t[u/x])\in\interp E$.
      \item $\pi_{E}((\lambda x^A.t)u)\toreq\pi_{E}(t[u_1/x]u_2)$, with
        $u\eq^*u_1\times u_2$, hence by Lemmas~\ref{lem:generation} and
        \ref{lem:unicity}, $C\equiv A\wedge C'$, with $u_1:A$ and $u_2:C'$. Therefore, by Lemma~\ref{lem:ImpImpSamePremise}, $B\equiv C'\Rightarrow(E\wedge F)$.
        Since $u\in\interp C$, we have $\pi_A(u)\in\interp A$ and
        $\pi_{C'}(u)\in\interp{C'}$. Then, by the i.h.
        \titre{CR2}, $u_1\in\interp A$ and $u_2\in\interp{C'}$. Then,
        $t[u_1/x]\in\interp B=\interp{C'\Rightarrow(E\wedge F)}$, so
        $t[u_1/x]u_2\in\interp{E\wedge F}$, so $\pi_E(t[u_1/x]u_2)\in\interp E$.
      \item $\pi_E((\lambda x^A.t)u)\toreq\pi_E((\lambda x^A.t_1)u\times(\lambda
        x^A.t_2)u)$, with $t\eq^*t_1\times t_2$. Hence, by
        Lemmas~\ref{lem:generation} and \ref{lem:unicity},
        $B\equiv B_1\wedge B_2$, with $t_1:B_1$, $t_2:B_2$.
        Since $t\in\interp{B}=\interp{B_1\wedge B_2}$, then
        $\pi_{B_i}(t)\in\interp{B_i}$, and so, by the i.h. \titre{CR2},
        $t_i\in\interp{B_i}$. In the same way, since
        $t[r/x]\in\interp B$, $t_i[r/x]\in\interp{B_i}$.
        Since $(A\Rightarrow B_1)\wedge (A\Rightarrow B_2)\equiv C\Rightarrow D$, we
        have, by Lemma~\ref{lem:ImpConj}, $D\equiv D_1\wedge D_2$, and $A\Rightarrow
        B_i\equiv C\Rightarrow D_i$.
        Then, by the i.h., $(\lambda x^A.t_1)u\in\interp{D_1}$ and
        $(\lambda x^A.t_2)u\in\interp{D_2}$. Therefore, since $m(D_1\times
        D_2)=m(D)<m(C\Rightarrow D)=m(T)$, by the i.h.
        \titre{Adequacy of product}, $(\lambda
        x^A.t_1)u\times(\lambda x^A.t_2)u\in\interp{D_1\wedge D_2}=\interp{D}=\interp{E\wedge
          F}$, so, by definition, $\pi_E((\lambda
        x^A.t_1)u\times(\lambda x^A.t_2)u)\in\interp E$.
        \qedhere
      \end{itemize}
    \end{enumerate}
  \end{itemize}
\end{proof}

We finally prove the adequacy lemma and the strong normalization theorem.

\begin{definition}[Adequate substitution]
  A substitution $\sigma$ is adequate if for all $x\in\V_A$,
  $\sigma(x)\in\interp{A}$.
\end{definition}

\begin{lemma}[Adequacy]\label{lem:adequacy}
  If $ r:A$, then for all $\sigma$ adequate, $\sigma r\in\interp A$.
\end{lemma}
\begin{proof}
  By induction on $r$. 
  \begin{itemize}
  \item If $r$ is a variable $x\in\V_A$, then, since $\sigma$ is adequate, we
    have $\sigma r\in\interp A$.
    
  \item If $r$ is a product $s\times t$, then by Lemma~\ref{lem:generation},
    $s:B$, $t:C$, and $A\equiv B\wedge C$, then by the i.h., $\sigma s\in\interp B$
    and $\sigma t\in\interp C$. By Lemma~\ref{lem:tout} (adequacy of product),
    ${(\sigma s\times\sigma t)}\in\interp{B\wedge C}$, hence, $\sigma r\in\interp
    A$.
    
  \item If $r$ is a projection $\pi_A(s)$, then by Lemma~\ref{lem:generation},
    $s:A\wedge B$, and by the i.h., $\sigma s\in\interp{A\wedge B}$. Therefore,
    $\sigma(\pi_A(s))=\pi_A(\sigma s)\in\interp A$.
    
  \item If $r$ is an abstraction $\lambda x^B.s$, with $s:C$, then by
    Lemma~\ref{lem:generation}, $A\equiv B\Rightarrow C$, hence by the i.h., for all
    $\sigma$, and for all $t\in\interp B$, $(\sigma s)[t/x]\in\interp C$. Hence, by
    Lemma~\ref{lem:tout} (adequacy of abstraction), ${\lambda x^B.\sigma
      s}\in\interp{B\Rightarrow C}$, hence, $\sigma r\in\interp A$.
    
  \item If $r$ is an application $st$, then by Lemma~\ref{lem:generation},
    $s:B\Rightarrow A$ and $t:B$, then by the i.h., $\sigma s\in\interp{B\Rightarrow
      A}$ and $\sigma t\in\interp B$. Then $\sigma(st)=\sigma s\sigma t\in\interp A$.
    \qedhere
  \end{itemize}
\end{proof}

\begin{theorem}[Strong normalization]\label{thm:SN}
  If $ r:A$, then $ r\in\SN$.
\end{theorem}
\begin{proof}
  By Lemma~\ref{lem:tout} \titre{CR3'}, for all type $B$, $x^B\in\interp B$, so
  the identity substitution is adequate. Thus, by Lemma~\ref{lem:adequacy} and
  Lemma~\ref{lem:tout} \titre{CR1}, $r\in\interp{A}\subseteq\SN$.
\end{proof}

\section{Consistency}\label{sec:cons}
We say that a term is $\basicr$-normal whenever it cannot continue reducing by
relation $\basicr$, that is, a term that cannot be $\beta$, $\pi$, or
$\zeta$-reduced, but may be expanded by rules $\eta$ or $\delta$.

\begin{lemma}\label{lem:consistencyOfProd}
  If $r:A\wedge B$ is closed $\basicr$-normal, then $r\eq^* r_1\times r_2$, with
  $r_1:A$ and $r_2:B$.
\end{lemma}
\begin{proof}
  We proceed by induction on $M(r)$.
  \begin{itemize}
  \item $r$ cannot be a variable, since it is closed.
  \item If $r=u\times v$, then by Lemma~\ref{lem:generation}, $u:C$, $v:D$, and
    $C\wedge D\equiv A\wedge B$. Then, by Lemma~\ref{lem:eqConjN}, one of the
    following cases happens
    \begin{itemize}
    \item $A\equiv C_1\wedge D_1$ and $B\equiv C_2\wedge D_2$, with $C\equiv
      C_1\wedge C_2$ and $D\equiv D_1\wedge D_2$. Then, by the i.h., $u\eq^*u_1\times
      u_2$ with $u_1:C_1$ and $u_2:C_2$, and $v\eq^*v_1\times v_2$ with $v_1:D_1$ and
      $v_2:D_2$. So, take $r_1=u_1\times v_1$ and $r_2=u_2\times v_2$.
    \item $B\equiv C\wedge D_2$, with $D\equiv A\wedge D_2$. Then, by the i.h.,
      $v\eq^*v_1\times v_2$. Take $r_1=v_1$ and $r_2=u\times v_2$. Three other cases
      are symmetric.
    \item $A\equiv C$ and $B\equiv D$, take $r_1=u$ and $r_2=v$. The last case
      is symmetric.
    \end{itemize}
  \item If $r=\lambda x^C.r'$, then, by Lemma~\ref{lem:generation}, $A\wedge
    B\equiv C\Rightarrow D$, and so, by Lemma~\ref{lem:ImpConj}, $D\equiv D_1\wedge
    D_2$, with $A\equiv C\Rightarrow D_1$ and $B\equiv C\Rightarrow D_2$. Hence, by
    the i.h., $r'\eq^*r'_1\times r'_2$ with $r'_1:D_1$ and $r'_2:D_2$. Therefore,
    $r\eq^*(\lambda x^C.r'_1)\times (\lambda x^C.r'_2)$, with $\lambda
    x^C.r'_1:C\Rightarrow D_1\equiv A$ and $\lambda x^C.r'_2:C\Rightarrow D_2\equiv
    B$.
  \item If $r=r_1r_2$, then by Lemma~\ref{lem:generation}, $r_1:C\Rightarrow
    A\wedge B\equiv (C\Rightarrow A)\wedge(C\Rightarrow B)$, so, by the i.h.,
    $r_1\eq^* s\times t$, and so $(s\times t)r_2\re sr_2\times tr_2$, so $r$ is not
    $\basicr$-normal.
  \item If $r=\pi_{A\wedge B}(r')$, then, by Lemma~\ref{lem:generation},
    $r':A\wedge B\wedge C$, so, by the i.h., $r'\eq^* s_1\times s_2$, with
    $s_1:A\wedge B$, and so $r$ is not $\basicr$-normal.
    \qedhere
  \end{itemize}
\end{proof}

\begin{theorem}[Consistency]\label{thm:consistency}
  There is no closed term in normal form of type $\tau$.
\end{theorem}
\begin{proof}
  Consider a closed term in normal form $r$ of type $\tau$.
  \begin{itemize}
  \item If $r$ is a variable, it is not closed.
  \item If $r$ is an abstraction or a product, then by
    Lemma~\ref{lem:generation}, it does not have type $\tau$.
  \item If $r$ is a projection $r=\pi_\tau(r')$, then, by
    Lemma~\ref{lem:generation}, $r':\tau\wedge A$. Hence, since $r$ is in normal
    form, $r'$ is $\basicr$-normal, so, by Lemma~\ref{lem:consistencyOfProd},
    $r'\eq^*r_1\times r_2$ with $r_1:\tau$, hence $r$ is not in normal form.
  \item If $r$ is an application, $r=st_1\dots t_n$, with $n\geq 1$, and
    $s\not\eq^*s_1s_2$, then let $t=t_1\times\cdots\times t_n$, so we have
    $r\eq^*st$, and consider the cases for $s$.
    \begin{itemize}
    \item $s$ cannot be a variable, since the term is closed.
    \item $s$ cannot be an abstraction $\lambda x^C.s'$, since, by
      Lemmas~\ref{lem:generation} and \ref{lem:ImpImp}, $t:C$, or $t:C\wedge D$. In
      the first case, the term $r$ is a $\beta$-redex, hence it is not in normal form,
      in the second case, we have that since $r$ and $t$ are in normal form, so it is
      also $\basicr$-normal, and by Lemma~\ref{lem:consistencyOfProd}, $t\eq^*u\times
      v$, with $u:C$, so $r\eq^* (\lambda x^C.s')uv$, which contains a $\beta$-redex.
    \item $s$ cannot be an application, by hypothesis.
    \item $s$ cannot be a product, since $st$ would be a $\zeta$-redex.
    \item $s$ cannot be a projection $\pi_A(s')$, since in such a case, by
      Lemma~\ref{lem:generation}, $s':A\wedge B$, and it would be $\basicr$-normal,
      so, by Lemma~\ref{lem:consistencyOfProd}, $s'\eq^*s_1\times s_2$ with $s_1:A$,
      and so, $r$ would contain a $\pi$-redex.
      \qedhere
    \end{itemize}
  \end{itemize}
\end{proof}

Note that, in the proof of Theorem \ref{thm:consistency}, we need
Lemma \ref{lem:consistencyOfProd} to handle the case of the
projection, but no analogous lemma for implication is needed, as both
$(\lambda x^A . s) t$ and $(s_1 \times s_2) t$ can be reduced.

\begin{theorem}[Introduction property]\label{thm:intro}
    ~
\begin{itemize}
\item  If $r$ is a closed term in normal form
of type $A \Rightarrow B$, then $r$ is an introduction.

\item  If $r$ is a closed term in normal form
of type $A \wedge B$, then $r$ is an introduction.
\end{itemize}
\end{theorem}

\begin{proof}
If $r$ has type $A \Rightarrow B$ and it is not an introduction
then it can be $\eta$-expanded and it is not normal.
If $r$ has type $A \wedge B$ and it is not an 
introduction it can be $\delta$-expanded and it is not normal.
\end{proof}

\begin{corollary}
  If $r:A$ is a closed term in normal form, then $r$ is an introduction.
\end{corollary}

\begin{proof}
  Since $r$ is a closed term in normal form, by
  Theorem~\ref{thm:consistency}, $A\neq\tau$. Hence $A=B\Rightarrow C$
  or $A=B\wedge C$. We conclude with Theorem \ref{thm:intro}.
\end{proof}

\section{Conclusion}\label{sec:conclusion}
We have proposed a calculus, \OCe, where conjunction is associative
and commutative, where implication distributes over conjunction and
where currified and uncurrified proofs are equated. In this calculus,
reduction is non-deterministic, but it enjoys termination and subject
reduction: a proof of a proposition $A$ always reduces to a proof of
this same proposition.

Compared with \OC, a first version of this calculus without the
extensionality rules, \OCe, also enjoys the introduction property, and
abstractions are not restricted to prime types.  This means that,
unlike in simply typed lambda-calculus, where the $\eta$-rule can be
considered or not, when isomorphic types are equated, this rule seems
to be mandatory to unblock terms like $(\lambda x^A.\lambda y^B.x)r$,
where $r:B$, $(\lambda x^{A\wedge B}.x)r$, where $r:A$, or
$\pi_{A\Rightarrow B}(\lambda x^A.r)$, where $r:B\wedge C$.

In this preliminary work, we consider implication and conjunction only.
This system needs to be extended to other connectives and quantifiers
of predicate logic and, possibly, to more complex systems, such as
dependent type theory. A first step in this direction is the
Polymorphic \OC \cite{PSI},
that adds the universal quantifiers at the level of types to \OC.

Yet, with these two connectives the proofs are more complex than for
simply typed lambda-calculus, but the work on Polymorphic \OC shows
that they scale, at least for the case of the universal quantifiers at
the level of types.

Finally, we have addressed in this paper the syntactic properties of
\OCe only.  The construction of a model for this system is left for
future work.

\section*{Acknowledgements}

The authors would like to thank Jean-Baptiste Joinet for useful comments and discussions.

\bibliographystyle{abbrvnat}
\bibliography{biblio}

\end{document}